\documentclass[10pt,journal,compsoc]{IEEEtran}
\usepackage{epsfig}
\usepackage{xspace}
\usepackage{url}
\usepackage{subfigure}
\usepackage{pgf}
\usepackage{pgfplots,pgfplotstable} 
\usepgfplotslibrary{units}

\usepackage{mathrsfs}
\usepackage{filecontents}
\usepackage{tikz} 
\usetikzlibrary{arrows,shapes,automata,positioning,decorations,calc} 
\usetikzlibrary{spy,backgrounds}
\pgfdeclaredecoration{complete sines}{initial}
{
    \state{initial}[
        width=+0pt,
        next state=sine,
        persistent precomputation={\pgfmathsetmacro\matchinglength{
            \pgfdecoratedinputsegmentlength / int(\pgfdecoratedinputsegmentlength/\pgfdecorationsegmentlength)}
            \setlength{\pgfdecorationsegmentlength}{\matchinglength pt}
        }] {}
    \state{sine}[width=\pgfdecorationsegmentlength]{
        \pgfpathsine{\pgfpoint{0.25\pgfdecorationsegmentlength}{0.5\pgfdecorationsegmentamplitude}}
        \pgfpathcosine{\pgfpoint{0.25\pgfdecorationsegmentlength}{-0.5\pgfdecorationsegmentamplitude}}
        \pgfpathsine{\pgfpoint{0.25\pgfdecorationsegmentlength}{-0.5\pgfdecorationsegmentamplitude}}
        \pgfpathcosine{\pgfpoint{0.25\pgfdecorationsegmentlength}{0.5\pgfdecorationsegmentamplitude}}
}
    \state{final}{}
}
\usepackage[latin1]{inputenc}

\ifCLASSOPTIONcompsoc
  \usepackage[nocompress]{cite}
\else
  \usepackage{cite}
\fi
\ifCLASSINFOpdf
\else
\fi
\usepackage[cmex10]{amsmath}
\usepackage{amsfonts,amsthm}
\usepackage{algorithmic, algorithm}
\usepackage{fancyvrb}
\usepackage{array}
\usepackage{multirow}

\newcommand{\SHArule}[1]{{\texttt{Rule~#1}\xspace}}

\newtheorem{lemma}{Lemma}

\newtheorem{definition}{Definition}
\newcommand{\am}[1]{{\textsf{\color{red} \small{[#1-am-]}}}}
\newcommand{\pr}[1]{{\textsf{\color{red} \small{[#1-pr-]}}}}
\newcommand{\sa}[1]{{\textsf{\color{red} \small{[#1-bb-]}}}}
\newcommand{\ey}[1]{{\textsf{\color{red} \small{[#1-ey-]}}}}
\newcommand{\mt}[1]{{\textsf{\color{red} \small{[#1-mt-]}}}}
\newcommand{\blue}[1]{{\color{blue} #1}}

\renewcommand{\am}[1]{}
\renewcommand{\pr}[1]{}
\renewcommand{\sa}[1]{}
\renewcommand{\ey}[1]{}
\renewcommand{\mt}[1]{}
\usepackage{enumitem}

\usepackage{xparse}
\usepackage{ifthen} 

\usepackage{acronym}

\newcommand{\ignore}[1]{{}}
\newcommand{\ourTool}{Piha\xspace}
\newcommand{\iSignal}{\tau} 
\newcommand{\inSignal}[1]{#1~} 
\newcommand{\outSignal}[1]{#1~} 

\newcolumntype{L}[1]{>{\raggedright\let\newline\\\arraybackslash\hspace{0pt}}m{#1}}
\newcolumntype{C}[1]{>{\centering\let\newline\\\arraybackslash\hspace{0pt}}m{#1}}
\newcolumntype{R}[1]{>{\raggedleft\let\newline\\\arraybackslash\hspace{0pt}}m{#1}}
\begin{document}
	
\acrodef{CPS}{Cyber-physical Systems}
\acrodef{DSP}{Digital Signal Processor}
\acrodef{DTTS}{Discrete Time Transition System}
\acrodef{EA}{Evolutionary Algorithm}
\acrodef{HA}{Hybrid Automata}
\acrodef{ILP}{Integer Linear Programming}
\acrodef{MCU}{Microcontroller Unit}
\acrodef{ODE}{Ordinary Differential Equation}
\acrodef{PoC}{Plant-on-a-Chip}
\acrodef{WHA}{Well-formed Hybrid Automata}
\acrodef{SHA}{Synchronous Hybrid Automata}
\acrodef{SWA}{Synchronous Witness Automata}
\acrodef{WCET}{Worst-Case Execution Time}
\newcommand*{\FIGUREwaterHeaterHA}{
\begin{tikzpicture}[->,>=stealth',shorten >=1pt,auto,node distance=5.4cm,
semithick,scale=0.9, transform shape]
\tikzstyle{every state}=[rectangle,rounded corners, minimum height = 2.2cm, text width=2.2cm, text centered, fill=blue!20,draw=none,text=black, draw,line width=0.3mm]

\node[initial,state, 
label={[shift={(0,-2.9)}]$ x = 20$}, 
label={[shift={(-1.5,-0.2)}]\large $ \mathbf{t_1}$ }]
(T1)  {$\dot{x}=0$};
\node[state, 
label={[shift={(0,0.1)}]$20\leq x \leq 100$}, 
label={[shift={(-1.5,-0.2)}]\large $\mathbf{t_2}$   }]
(T2) [above of=T1] 	 {$\dot{x}=K(h-x)$};
\node[state, 
label={[shift={(0,0.1)}]$ x = 100$}, 
label={[shift={(1.5,-0.2)}]\large $\mathbf{t_3}$  }] 
(T3) [right of=T2] 	 {$\dot{x}=0$};
\node[state, 
label={[shift={(0,-2.9)}]$20\leq x \leq 100$}, 
label={[shift={(1.5,-0.2)}]\large $\mathbf{t_4}$  }]
(T4) [right of=T1] 	 {$\dot{x}=-Kx$};

\draw (T1) -- (T2) node [midway] {\footnotesize $ON,x^\prime = x$};
\draw (T2) -- (T3) node [midway] {\footnotesize $B,x=100 \wedge x^\prime = x$};
\draw (T3) -- (T4) node [midway, right] {\footnotesize $OFF,x^\prime = x$};
\draw (T4) -- (T1) node [midway] {\footnotesize $C,x=20 \wedge x^\prime = x$};

\draw (T2.330) -- (T4) node [midway, sloped] {\footnotesize $OFF,x^\prime = x$};
\draw (T4.150) -- (T2) node [midway, sloped] {\footnotesize $ON,x^\prime = x$};

\end{tikzpicture}}

\newcommand{\iflame}[2]{%
  \begin{tikzpicture}[scale=#1]
    \fill[color=#2]
      (0,0) .. controls (-1.5,1.25) and (.5,2) .. (-.2,4)
            .. controls (1,2.5) and (1,.5) .. (0,0);
  \end{tikzpicture}
}

\newcommand{\icandle}[1]{%
  \rlap{\kern-.0275em\raisebox{1.2ex}{\iflame{.20}{#1}}}%
}

%
\title{A synchronous rendering of hybrid systems for designing
  Plant-on-a-Chip (PoC)}
%
%
%
%

\author{Avinash~Malik,
        Partha~S~Roop,
        Sidharta~Andalam, Eugene~Yip,
        and Mark~Trew
\IEEEcompsocitemizethanks{\IEEEcompsocthanksitem Authors are with the Department
of Electrical and Computer Engineering, and  Auckland Bioengineering Institute at
The University of Auckland, New Zealand\protect\\
E-mail: avinash.malik@auckland.ac.nz
\IEEEcompsocthanksitem}
\thanks{Manuscript received xx xxxx xxxx}}

%
%

\markboth{ IEEE Transactions on Software Engineering,~Vol.~, No.~, XXXX~XXXX}%
{}
%



\IEEEtitleabstractindextext{%
\begin{abstract}
  Hybrid systems are discrete controllers that are used for
  controlling a physical process (plant) exhibiting continuous dynamics.
  A hybrid automata (HA) is a well known and widely used formal model for the
  specification of such systems. While many methods exist for
  simulating hybrid automata, there are no known approaches for the
  automatic code generation from HA that are semantic
  preserving. If this were feasible, it would enable the design of
  a plant-on-a-chip (PoC) system that could be used for the \emph{emulation} of the
  plant to validate discrete controllers. Such an approach would need to be
  mathematically sound and should not rely on
  numerical solvers.
   We propose a method of PoC design for plant emulation, 
   not possible before.  The approach
  restricts input/output (I/O) HA models using a set 
  of criteria for well-formedness which are statically verified.  
  Following verification, we use an abstraction based on a 
  synchronous approach to facilitate code generation.  
  This is feasible through a
  sound transformation to synchronous HA.
   We compare our method  (the developed tool called \ourTool) 
    to the widely used Simulink\textsuperscript{\textregistered}
   simulation framework 
  and show that our method is superior in both execution time and code size. 
  Our approach to the PoC problem paves the way for the emulation 
  of physical plants in diverse domains such as robotics, automation,
  medical devices, and intelligent transportation systems.
\end{abstract}


}

\maketitle

\IEEEdisplaynontitleabstractindextext

%
\IEEEpeerreviewmaketitle
\acresetall	

\IEEEraisesectionheading{\section{Introduction}
\label{sec:introduction}}
\acresetall 

\ac{CPS}~\cite{lee08cps, alur2015principles} encompass a wide range of
systems where distributed embedded controllers are used for controlling
physical processes. Examples range from the controller of a robotic arm,
controllers in automotive such as drive-by-wire applications to
pacemakers for a heart. Such systems, also called hybrid systems,
combine a set of discrete controllers with associated physical plants
that exhibit continuous dynamics. An individual plant combined with its
controller is often formally described using
\ac{HA}~\cite{raskin05},~\cite{ alur2015principles, alur93}. These \ac{HA} use a
combination of discrete \emph{locations} or \textit{modes} and \acp{ODE}
to capture the continuous dynamics of an individual mode.

\ac{CPS} pose many key challenges for their design, especially since
they are inherently safety-critical. CPS must operate in a sound manner
while preserving key requirements for both functional and timing
correctness~\cite{wilhelm08wcrt}. Moreover, there is a need for most
\ac{CPS} to be certified using functional safety standards such as
IEC~61580~\cite{iec61508} (robotics and automation),
ISO~26262~\cite{iso26262} (automotive), DO-178B~\cite{youn2015software}
(flight control), US Food and Drug Administration (FDA)~\cite{fda}
(medical devices). Certification is a key barrier to rapid adoption of
new technology as compliance costs are prohibitive. Also, in spite of
products being certified, faults can occur and associated recalls are
required, which can be immensely costly in terms of human lives and
economics. For example, close to 200,000 pacemakers were recalled
between 1990--2000 due to software related failures and this trend is
continuing~\cite{alemzadeh13}.

In spite of these safety crtical needs, current design practices often
use tools without rigorous semantics. For example, simulation tools such
as Matlab\textsuperscript{\textregistered} Simulink\textsuperscript{\textregistered} and
Stateflow\textsuperscript{\textregistered}~\cite{alur08symbolic, bourke13zelus} are arguably
the default standard in many fields, including automotive and medical.
These tools enable the encoding of an HA using a set of
Simulink\textsuperscript{\textregistered} blocks to specify the continuous dynamics while the
discrete aspects are modelled using Stateflow\textsuperscript{\textregistered}.  Due to the
lack of rigorous semantics in Simulink\textsuperscript{\textregistered}/Stateflow\textsuperscript{\textregistered},
many problems have been reported in literature regarding the lack of
good model fidelity. To overcome these problems, tools with rigorous
semantics such as Ptolemy~\cite{ptolemaeus2014system} and
Z\'elus~\cite{bourke15code, bourke13zelus} have been developed. These
tools have developed mathematical semantics for the composition of the
blocks and hence do not have semantic problems. However, all available
tools, we are aware of, use numerical solvers during simulation to solve
the ODEs. This poses two key challenges in the design of the
CPS. Firstly, the simulation is inherently slow and potentially
unsuitable for validating controllers in closed loop that can be
critical to certification. A second problem is that simulation fidelity
is dependent on the nature of the \ac{ODE} solver and changing the
solver may result in a different outcome.

Hardware-in-the loop simulation, also known as {\em
  emulation}~\cite{patel2015survey}, uses the actual plant in
conjunction with the controller for validation.  Emulation has been
shown in many domains (such as motor controllers) to be an extremely
effective method of validation. However, emulation of a hybrid system
may not be feasible due to the following reasons. (1)~The physical
process may also be part of the design and not available for testing,
e.g. a robotic arm and its controller are required but the actual robot
may not exist yet. This is an even more significant hurdle for medical
devices such as pacemakers, where emulating the actual system is often
not feasible or may require animal experiments that are challenging in
terms of ethics, time and cost.  (2)~Simulation of physical process
using existing tools may not be able to meet the timing requirements of
emulation, e.g. real time responses.

Our work overcomes these barriers by proposing, for
the first time, the
concept of plant-on-a-chip (PoC). 
A PoC is an implementation of a hybrid
system in C (and subsequently VHDL), so that a physical process
can be emulated using a processor, ASIC, or FPGA. This device can then
be used for closed-loop validation of controllers in wide ranging
applications in CPS such as robotics, automotive, and medical devices.

\begin{figure}[ht]
\centering

\begin{tikzpicture}[->,>=stealth',shorten >=1pt,auto,
node distance=3.7cm,
semithick,scale=0.9, transform shape]
\tikzstyle{every state}=[rectangle,rounded corners, minimum height = 1.0cm, text width=2.5cm, text centered,draw=none,text=black, draw,line width=0.3mm, font=\footnotesize, inner xsep=0cm, inner ysep=0.2cm ]

\node[state,fill=blue!10]
     (S)  {Specification \\ (step~1)};
     
\node[state,fill=blue!15]
      (W) [right of=S] {Well-formed \\ hybrid automata \\ verification \\ (step~2)};
      
\node[state,fill=blue!20]
      (M) [right of=W] {Modular \\ compilation \\  (step~3)};

\node[state,fill=blue!25, node distance=2.2cm]
      (L) [below of=M] {Linking \\ parallel \\ composition\\ (step~4)};

\node[state,fill=blue!30, node distance=2.2cm]
      (C) [below of=W] {C code and \\ math library\\ (step~5)};

\node[state,fill=blue!35, node distance=2.2cm]
      (D) [below of=S] {Deployment \\ on the \\ embedded device\\ (step~6)};
      
\draw[ultra thick] (S) -- (W);
\draw[ultra thick] (W) -- (M) node [midway,above, pos=0.50] {\footnotesize Pass};
\draw[ultra thick] (W.90) -- ([yshift=1em]W.90) -| (S.90)
	node [midway, above, pos=0.25] 
	{\footnotesize Fail--review the specification};
\draw[ultra thick] (M) -- (L);
\draw[ultra thick] (L) -- (C);
\draw[ultra thick] (C) -- (D);
           
\end{tikzpicture}

\caption{An overview of the proposed methodology}
\label{fig:methodology}

\end{figure}
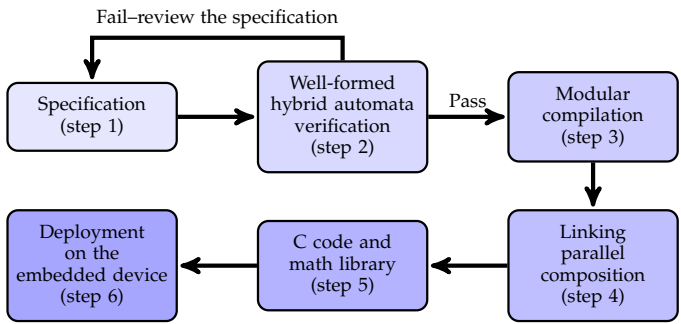

The proposed approach significantly extends the synchronous
approach~\cite{benveniste03} that is the default standard for designing
safety critical systems in many domains including civil
aviation. Consequently, our code generation relies on the well known
tenets of the synchronous approach. An overview
of the proposed methodology is presented in
Figure~\ref{fig:methodology}.  We start with a specification that
describes a network of HA models of a CPS (step 1). To facilitate code
generation, we constrain such models using a set of well formedness
criteria, which can be statically verified (step 2). We then introduce
an abstraction called synchronous hybrid automata (SHA), that provides
the semantics of the generated code, and using this abstraction separate
compilation can be performed for each HA (step 3). Such modular
compilation is feasible due to the synchronous semantics of a network of
well formed HAs (WHAs) presented in Section~\ref{sec:DTTS-par}. Using
this semantics, we perform a linking operation of the generated C code
(step 4). Finally, the generated C code (step 5) can be compiled using a
standard C compiler to deploy the binary on an embedded target platform
(step 6). The proposed methodology is detailed in
Section~\ref{sec:code-gen}.

\subsection{Related work}
\label{sec:related}

The design of CPS requires both formal analysis~\cite{baier08} and
testing-based validation. Of particular importance is the need for
emulation-based validation~\cite{patel2015survey}, as outlined above.
Hybrid automata and their compositions~\cite{alur93, raskin05} provide
a formal framework for the specification of hybrid systems. There are
well known negative results on the decidability of the reachability
problem 
 over hybrid automata, i.e., even the
simplest class, called \emph{linear hybrid automata}, is
undecidable~\cite{raskin05}. 
Several restrictions have been proposed
with associated semi-decidable results that are used for developing
model checking-based solutions~\cite{SpaceEx}.
 However, the negative decidability
results mean that both formal analysis and modular code generation are critical~\cite{bourke13zelus}.

Unlike the objective of modular code generation based on formal models,
conventional \ac{CPS} design relies on commercial tools such as
Simulink\textsuperscript{\textregistered} or Stateflow\textsuperscript{\textregistered}~\cite{alur08symbolic}.  However, 
Simulink\textsuperscript{\textregistered} and Stateflow\textsuperscript{\textregistered} lack formal semantics for
the composition of the components leading to issues with simulation
fidelity (see~\cite{tripakis05translating, bourke13zelus} for
details). In contrast to commercial tools, academic variants such
as Ptolemy~\cite{ptolemaeus2014system} and Z\'elus~\cite{bourke13zelus}
are founded on formal models of computation. Ptolemy supports multiple
models of computations and a specification can seamlessly combine
them. This approach is extremely powerful for the simulation of complex
\ac{CPS}.  These tools can model hybrid systems that use numerical solvers.

Sequential code generation from synchronous specifications is a very
well developed area~\cite{benveniste03, andalam14}. For example, the
code generator from SCADE
specifications~\cite{scade,fornari2010understanding} generates certified
code according to the DO-178B standard.  However, existing synchronous
code generators are designed for discrete systems only and unsuitable
for generating code from HA descriptions.  Z\'elus~\cite{bourke13zelus}, on the other hand,
is an extension of the synchronous language Lustre~\cite{benveniste03}
with \acp{ODE} for the design of hybrid systems.  Z\'elus~\cite{bourke13zelus} overcomes the
semantic ambiguities associated with Simulink\textsuperscript{\textregistered}. While Ptolemy~\cite{ptolemaeus2014system} and Z\'elus~\cite{bourke13zelus}
are significant developments for CPS design due to their sound
semantics, they are limited by their reliance on dynamic \ac{ODE}
solvers during simulation. Consequently, they are unsuitable for the
emulation of physical processes. However, we note that approaches such as 
Ptolemy~\cite{ptolemaeus2014system} and Z\'elus~\cite{bourke13zelus} are highly
expressive as they can simulate any HA, 
while the proposed approach generates code 
from \ac{WHA} models, which are a subset of HA.

\ignore{Here we propose an approach for validation based on closed-loop
  emulation of the CPS using a PoC to emulate the behaviour of the
  physical processes accurately. This requires code generation from
  hybrid automata models, unlike conventional approaches that use
  numerical simulators to solve ODES dynamically. This is essential due
  to two key reasons. Firstly, ODE solvers operate very speculatively
  (including backtracking steps) and changing an ODE solver may change
  the outcome of simulation. Secondly, the invocation of an dynamic ODE
  solver may substantially slow-down the execution, which is not ideal
  for embedded implementation of the generated code to facilitate
  real-time, closed-loop emulation.}

Code generation without using dynamic invocation of 
numerical solvers have been attempted in
the past with limited success.  Alur et al.
generated code from statechart-like hierarchical HA
models~\cite{alur2003generating}. Like our approach, the authors proposed a requirement of restricting
HA execution to what was statically verifiable. They enforced the requirement
that ``a given mode and the corresponding guard of any switch must
overlap for a duration that is greater than the sampling
period''~\cite{alur2003generating}. However, we show that code generation is
feasible without the need for such a restrictive approach. In a similar
vein, Kim et al. also presented an approach for modular code
generation from HA~\cite{kim2003modular}.  Their key idea was also based on the discretization of
HA, requiring that ``transitions are taken at time multiples
of the discretization step.'' Also, a SystemC based code
  generation for \ac{HA} has been proposed~\cite{Bresolin2011}. 
  While the title claims code generation for \ac{HA}, the actual
  algorithms are restricted to timed automata~\cite{alur94} and hence unsuitable
  for the general class of \ac{HA}.
  This is not adequate for
  modelling the physical plant properties in many examples.
   In
contrast, our approach  (the developed tool called \ourTool
\footnote{Piha is the name of a 
	well known New Zealand surfing paradise.
	Piha (Pi-Hybrid Automata) also refers to synchronous HA
	and is not related to the concept of $\pi$ in   pi-calculus.
	})
generates code using a technique that is similar
in principle, but without these overly restrictive requirements.  Hence,
we demonstrate that many practical \ac{CPS}, 
ranging from pacemakers to
boiler systems, can be designed and 
implemented using the proposed approach.

\subsection{Contributions and organisation}

The key contributions of the paper are:

\begin{itemize}
\item \emph{A software engineering (SE) approach for designing CPS}: We
  propose an approach that is based on key SE principles founded on
  formal methods and model-driven engineering~\cite{yoong2015model}. The
  overall system design starts with the description of the system as a
  network of hybrid automata. We propose a new semantics of HA based on
  the notion of \acf{WHA}. The WHA requirements can be
  statically checked and those that satisfy the WHA requirements are
  amenable to code generation.
\item \emph{Code generation without dynamic numerical solving}: We use
  a synchronous approach for the design of hybrid systems~\cite{bourke13zelus, bourke15code}.
  Unlike these code generators, which rely on
  dynamic ODE solvers rendering them less suitable for PoC design, we
  generate code that uses no numerical solver at run-time. This is
  achieved by an algorithm that checks WHA requirements at compile time
  to create closed form solutions of ODEs using symbolic solvers.
   WHA based code generation is less
   restrictive compared with existing code generators for
   HA~\cite{alur2003generating, kim2003modular} that also avoid dynamic
   ODE solvers.
\item \emph{Semantic preserving modular compilation of HA}: 
	The proposed approach is based on a compositional
	semantics such that code can be generated for each \ac{HA} separately.
\item \emph{A new PoC design approach}: We propose the first
  comprehensive approach for PoC design enabling the emulation of many CPS
  applications. Our approach, thus, paves the way for the reduction of
  the certification effort for safety-critical devices such as
  pacemakers~\cite{chen14}, automotive electronics, and many other
  applications where certification is a requirement.
\item \emph{Practicality}: We have applied our approach to the design of
  several hybrid systems and compared our implementation with the widely used 
  Simulink\textsuperscript{\textregistered} and Stateflow\textsuperscript{\textregistered}.
\item \emph{Software tool}: We have implemented these concepts in a software tool called Piha.
\end{itemize}

The paper is organised as follows. In section~\ref{sec:background} we
introduce HA as a modelling framework for CPS. We also define WHA as a
restrictive class that is amenable for code generation.  In
section~\ref{sec:synchronous} we introduce an abstraction of WHA called
synchronous hybrid automata and present its semantics.  \ignore{We also
formalise the soundness of this semantic abstraction.}  In
section~\ref{sec:code-gen} we present algorithms for code generation
based on the developed synchronous semantics.  The algorithms include
those to determine if an HA meets the requirements for code generation
based on WHA. It also includes the back-end code generators from HA
definitions to C using the developed synchronous semantics.
Section~\ref{sec: results} compares the proposed approach with existing
tools such as Simulink\textsuperscript{\textregistered} and Stateflow\textsuperscript{\textregistered}. Finally in
section~\ref{sec:conclusions} we make concluding remarks including the
limitations of the developed approach and future research avenues.


\section{Background: Hybrid automata}
\label{sec:background}
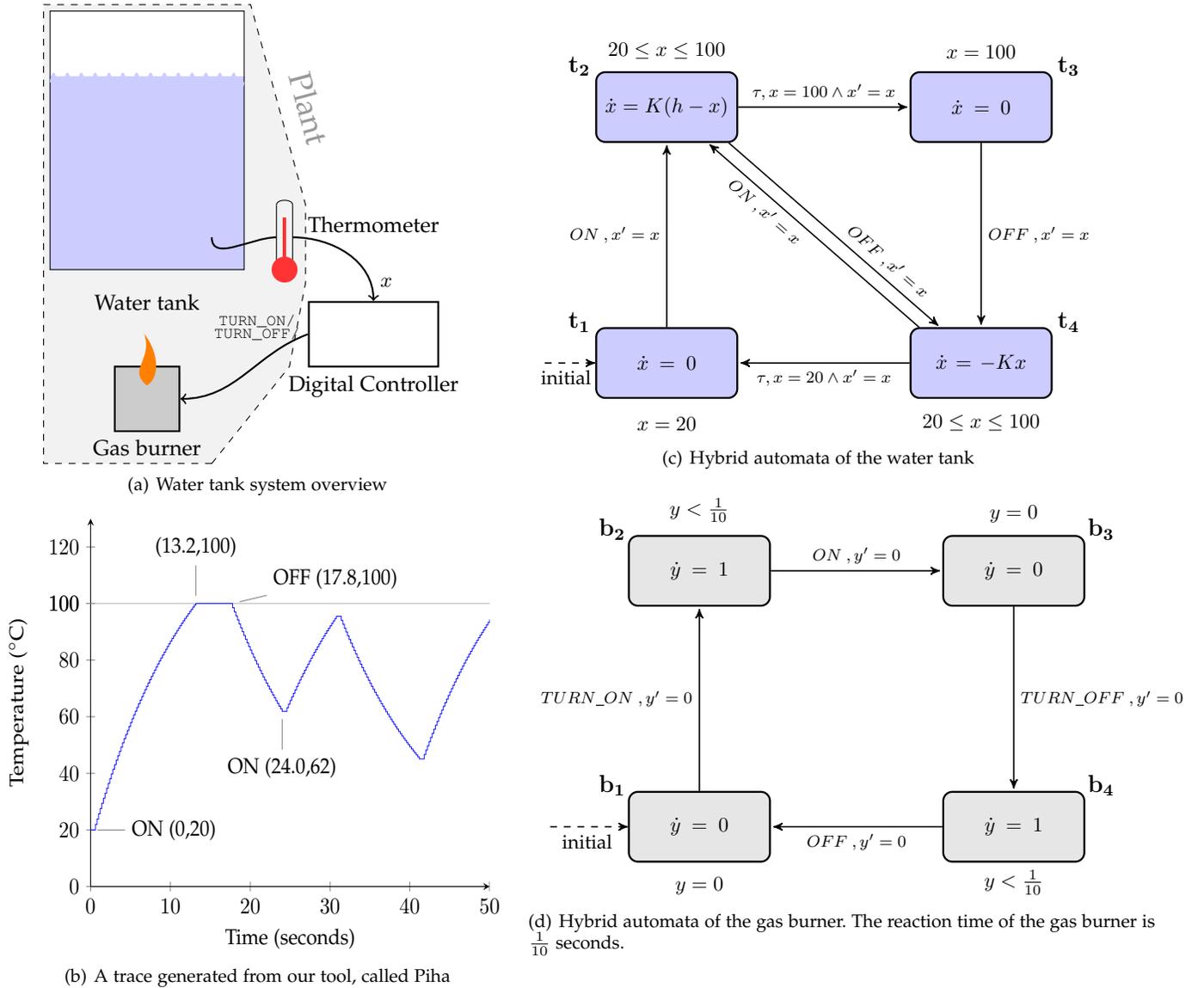
\begin{figure*}[ht]
	\begin{minipage}{0.45\linewidth}
		\centering
		\subfigure[Water tank system overview]{
			\begin{tikzpicture}

	\draw [dashed,  fill=gray!10] plot [smooth cycle, tension=0] 
	coordinates { (-0.1,-3.0) (2.5,-3) (3.7,-1.5)(3.95,-0.2) (3.95,1) (3.1,4.1) (-0.1,4.1) };
	\draw[gray!80] (4,2.5) node[rotate =-75] {\Large  Plant}; 

     \draw[fill=blue!20, draw=blue!20] 
     (0,0) rectangle (3,3);
     \draw[fill=white, draw=white] 
     (0,3) rectangle (3,4);
     \draw (0,0) rectangle (3,4);
     \draw (1.5,-0.5) node {Water tank}; 
    
    \draw[fill=blue!20, draw=blue!20, decorate, 
    decoration={complete sines,amplitude=8pt, segment length=6pt}
    ] 
         (0,2.9) -- (3,2.9);

    \node[fill=gray!40] (GB) at (1.5,-2) [draw,thick,minimum width=1cm,minimum height=1.0cm, label=below:Gas burner] {};
      
    \node (FLAME) at (GB.120) {\icandle{orange}};


	\draw (3.5,1)--(3.5,0);
	\draw (3.75,1)--(3.75,0);
	\draw (3.5,1) to[out=90,in=90] (3.75,1);
	\draw[fill=red!80, draw=none] 
	(3.625,0) circle (0.2);
	\draw[fill=red!80, draw=none] 
	(3.6,0) rectangle (3.65,0.8);
	\draw (5,0.7) node {Thermometer};

	\node (DC) at (5,-1) [draw,thick,minimum width=2cm,minimum height=1cm, label=below:Digital Controller] {};

	\draw[thick] (2.5,0.5) to[out=-90,in=180] (3.5,0.5);
	\draw[thick,->] (3.75,0.5) to[out=0,in=90] (DC.90);
	\draw (5.2,-0.2) node {\small \tt $x$};
	
	\draw[<-, thick] (GB) to[out=0,in=200] (DC.180);
	\draw (3.2,-0.8) node {\scriptsize \tt TURN\_ON/};
	\draw (3.2,-01) node {\scriptsize \tt TURN\_OFF/};

\end{tikzpicture}
			\label{fig:waterTankSystem}
		}
		\subfigure[A trace generated from our tool, called~\ourTool]{
			
			 \begin{tikzpicture}[transform shape, xscale=0.9]
\begin{axis}
[ xlabel={Time (seconds)},
ylabel={Temperature ({$^\circ$C})},
axis y line = left,
axis x line = bottom,
xmin=0,   xmax=50,
ymin=0,   ymax=130,
extra y ticks={100},
extra tick style={grid=major}
]
\addplot[color=blue!90,
mark=.,
mark size=2,
smooth,
const plot
]
table [x=b, y=c, col sep=comma] {./figures/WTdata.csv};
\node[pin=0:{ON (0,20)}] at (axis cs:0,20) {};
\node[pin=90:{ (13.2,100)}] at (axis cs:13.2,100) {};
\node[pin=10:{OFF (17.8,100)}] at (axis cs:17.8,100) {};
\node[pin=-90:{ON (24.0,62)}] at (axis cs:24.0,61.87) {};
\end{axis}
\end{tikzpicture}
			\label{fig:waterTankTrace}	
		}
	\end{minipage}
	\begin{minipage}{0.54\linewidth}
		\subfigure[Hybrid automata of the water tank]{
			\begin{tikzpicture}[->,>=stealth',shorten >=1pt,auto,
node distance=5.4cm,
semithick,scale=0.9, transform shape]
\tikzstyle{every state}=[rectangle,rounded corners,
 minimum height = 1.2cm, text width=2.2cm, text centered, fill=blue!20,draw=none,text=black, draw,line width=0.3mm]

\node[state, 
label={[shift={(0,-1.9)}]$ x = 20$}, 
label={[shift={(-1.5,-0.2)}]\large $ \mathbf{t_1}$ }]
(T1)  {$\dot{x}=0$};

\draw[<-, dashed](T1.180) -- node[below] {initial} ++(-1cm,-0cm);

\node[state, 
label={[shift={(0,0.1)}]$20\leq x \leq 100$}, 
label={[shift={(-1.5,-0.2)}]\large $\mathbf{t_2}$   }]
(T2) [node distance=4.4cm, above of=T1] 	 {$\dot{x}=K(h-x)$};

\node[state, 
label={[shift={(0,0.1)}]$ x = 100$}, 
label={[shift={(1.5,-0.2)}]\large $\mathbf{t_3}$  }] 
(T3) [right of=T2] 	 {$\dot{x}=0$};

\node[state, 
label={[shift={(0,-1.9)}]$20\leq x \leq 100$}, 
label={[shift={(1.5,-0.2)}]\large $\mathbf{t_4}$  }]
(T4) [right of=T1] 	 {$\dot{x}=-Kx$};

\draw (T1) -- (T2) node [midway] 
{\footnotesize $\inSignal{ON},x^\prime = x$};
\draw (T2) -- (T3) node [midway] 
{\footnotesize $\iSignal,x=100 \wedge x^\prime = x$};
\draw (T3) -- (T4) node [midway, right] 
{\footnotesize $\inSignal{OFF},x^\prime = x$};
\draw (T4) -- (T1) node [midway] 
{\footnotesize $\iSignal,x=20 \wedge x^\prime = x$};

\draw (T2.330) -- (T4) node [midway, sloped] 
{\footnotesize $\inSignal{OFF},x^\prime = x$};
\draw (T4.150) -- (T2) node [midway, sloped] 
{\footnotesize $\inSignal{ON},x^\prime = x$};

\end{tikzpicture}
			\label{fig:waterTankHAtank}

		}
		\subfigure[Hybrid automata of the gas burner. 
		The reaction time of the gas burner
		is \newline $\frac{1}{10}$ seconds.]{
			\begin{tikzpicture}[->,>=stealth',shorten >=1pt,auto,
node distance=5.4cm,
semithick,scale=0.9, transform shape]
\tikzstyle{every state}=[rectangle,rounded corners,
 minimum height = 1.2cm, text width=2.2cm, text centered, fill=gray!20,draw=none,text=black, draw,line width=0.3mm]

\node[state, 
label={[shift={(0,-1.9)}]$ y=0$}, 
label={[shift={(-1.5,-0.2)}]\large $ \mathbf{b_1}$ }]
(B1)  {$\dot{y}=0$};

\draw[<-, dashed](B1.180) -- node[below] 
{initial} ++(-1.4cm,-0cm);

\node[state, 
label={[shift={(0,0.1)}]$y< \frac{1}{10}$}, 
label={[shift={(-1.5,-0.2)}]\large $\mathbf{b_2}$   }]
(B2) [node distance=4.4cm, above of=B1] 	 {$\dot{y}=1$};

\node[state, 
label={[shift={(0,0.1)}]$ y=0$}, 
label={[shift={(1.5,-0.2)}]\large $\mathbf{b_3}$  }] 
(B3) [right of=B2] 	 {$\dot{y}=0$};

\node[state, 
label={[shift={(0,-1.9)}]$y< \frac{1}{10}$}, 
label={[shift={(1.5,-0.2)}]\large $\mathbf{b_4}$  }]
(B4) [right of=B1] 	 {$\dot{y}=1$};

\draw (B1) -- (B2) node [midway] 
{\footnotesize $\inSignal{TURN\_ON},y^\prime = 0$};
\draw (B2) -- (B3) node [midway] 
{\footnotesize $\outSignal{ON},y^\prime = 0$};
\draw (B3) -- (B4) node [midway, right] 
{\footnotesize $\inSignal{TURN\_OFF},y^\prime = 0$};
\draw (B4) -- (B1) node [midway] 
{\footnotesize $\outSignal{OFF}, y^\prime = 0$};

\end{tikzpicture}
			\label{fig:waterTankHAburner}	
		}
	\end{minipage}

	\caption{
		An example of a water tank heating system -- adapted from~\cite{raskin05}.
		Figure~\ref{fig:waterTankSystem} presents the plant and the discrete controller. 
		Figures~\ref{fig:waterTankHAtank} and ~\ref{fig:waterTankHAburner}  captures the behaviour of the water tank  and the gas burner.
		Figure~\ref{fig:waterTankTrace} presents an example trace through a water tank heating system (generated from our tool). The trace is a set of discrete points sampled every $0.2$ seconds.
		 }
	\label{fig:waterTank}
\end{figure*}



We use hybrid IO automata of Lynch et al.~\cite{lynch03} to capture the
behaviour of the hybrid system similar to the models adopted by Chen et
al.~\cite{chen14}.  This approach is specifically amenable to the
modelling of both the plant and its adjoining controller, and hence
ideal for the PoC design approach presented here.  We refer to them as
\acf{HA} for convenience.

To illustrate the concept of \ac{HA}, consider a
water tank heating system, presented in Figure~\ref{fig:waterTank}
(adapted from~\cite{raskin05}), where the temperature of the water
inside the tank is to be maintained between $20$ and $100^\circ$C. The
temperature is controlled by a digital controller that switches the
heating system \texttt{ON} and \texttt{OFF} at specific times so as to
maintain the temperature within the required range. The temperature of
water inside a tank is modelled using Newton's law of heating and
cooling, which can be represented as $x(t) = I e^{-Kt}+h(1-e^{-Kt})$.
Here, $I$ is the initial temperature, $K$ is a constant that depends on
the thermal conductivity of the tank, and $h$ is a constant that depends
on the power output of the heating system, such as a gas burner.  An
example trace of the temperature of water inside the tank for an initial
temperature of $20^\circ$C is shown in Figure~\ref{fig:waterTankTrace}.
 The initial temperature is $20^\circ$C
  and the heating system is turned \texttt{ON}.  Next, the temperature
  of water increases until it boils at $100^\circ$C, when the time
  instant is $13.2$ seconds.  Next, the temperature remains at
  $100^\circ$C until $17.8$ seconds, when the heating system is turned
  \texttt{OFF} by the controller.  The temperature of water drops until
  $24.6$ seconds, when the controller switches \texttt{ON} the heating
  system again.  We describe the water tank as a hybrid automata in
Figure~\ref{fig:waterTankHAtank} and that of the gas burner in
Figure~\ref{fig:waterTankHAburner}.

The specification of the water tank heating system using an HA is shown
in Figure~\ref{fig:waterTankHAtank}.  An HA captures the interaction
between the plant (exhibiting \emph{continuous dynamics}) and a
controller (making \emph{discrete mode switches}).  The continuous
evolution of some real-valued variables is modelled using a set of
\acfp{ODE}, which specify the rate of change of these variables. The
controller effects the discrete mode switches and in each mode the plant
exhibits different dynamics.  As a consequence, the ODEs in the
different modes (also called locations) may be different.  For the water
tank heating system, the HA can be in any one of four different
locations \texttt{$t_1$}, \texttt{$t_2$}, \texttt{$t_3$}, and
\texttt{$t_4$}, where \texttt{$t_1$} is the initial location.  The
variable $x$ is used to model the temperature of the water inside
the tank.  Each location has some flow predicates that specify the rate
of change of the continuous variables.  For example, in location $t_2$
the flow predicate that defines the rate of change of temperature when
the heater is turned \texttt{ON} is specified as $\dot{x}=K(h-x)$.
Invariants are associated with locations, e.g. $20 \leq x \leq 100$ is
an invariant associated with location $t_2$.  Execution remains in a
location until the invariant holds or an egress transition triggers
prior to this.  Some locations may have initialization conditions that
provide the initial values of the variables i.e. location $t_1$.  A
transition out of a locations is enabled when the input (event) is
present and the jump condition associated with the transition
holds. When a given transition is taken, the final value of the
variables are updated. For example, the transition from $t_2$ to $t_3$
is enabled when the value of temperature is $100^\circ$C and as the
transition is taken, the final value of $x$ is updated using the
equation $x'=x$, where $x'$ is an update variable. The values of the
update variables are used as the initial values of the variables in the
destination location.

An HA~\cite{lynch03} is defined using Definition~\ref{def:ha}. For the
water tank example shown in Figure~\ref{fig:waterTankHAtank},
$Loc=\{t_1, t_2, t_3, t_4\}$, $\Sigma=\{ON, OFF, \tau\}$, and an example
edge in $Edge$ is $(t_1, ON, t_2)$. There is a single continuous
variable $x$ representing the temperature of the water in the
tank. Hence, $X=\{x\}$, $\dot{X}=\{\dot{x}\}$, and $X'=\{x'\}$. The only
location that is marked initial is $t_1$ and $Init(t_1): x=20$. An
example flow predicate is $Flow(t_1): \dot{x}=0$, which specifies that
the temperature of the water inside the tank does not change.  A jump
predicate assigned to an edge specifies the condition needed to take a
transition and the updating of values of some continuous variables when
the transition is taken.  For example,
$Jump(t_4, \tau, t_1): x=20 \wedge x'=x$ specifies that the transition
is taken when the value of the temperature is $20^\circ$C and the
updated value of the temperature is also $20^\circ$C. We now formalise
the HA using Definition~\ref{def:ha}.

\begin{definition}
  A hybrid automata is \newline 
  $H = \langle Loc, Edge, \Sigma, Inv, Flow, Jump \rangle$ where:
  \begin{itemize}
  \item $Loc=\{l_1,..,l_n\}$ representing $n$ control modes or
    locations.
  \item $Edge \subseteq Loc \times \Sigma \times Loc$ are the set of
    edges between locations.
  \item $\Sigma=\Sigma_I \cup \Sigma_O \cup \{\iSignal\}$ is set of
    event names comprising of input, output, and internal events.
  \item Three sets for the set of continuous variables, their rate of
    change and their updated values represented as follows: \newline
    $X=\{x_1,.., x_m\}$, $\dot{X}=\{\dot{x_1},.., \dot{x_m}\}$, and \newline
    $X'=\{x_{1}',.., x_{m}'\}$.
  \item $Init(l)$: Is a predicate whose free variables are from $X$. It
    specifies the possible valuations of these when the HA starts in
    $l$.
  \item $Inv(l)$: Is a predicate whose free variables are from $X$ and
    it constrains these when the HA resides in $l$.
  \item $Flow(l)$: Is a predicate whose free variables are from
    $X \cup \dot{X}$ and it specifies the rate of change of these
    variables when the HA resides in $l$.
  \item $Jump(e)$: Is a function that assigns to the edge `$e$' a
    predicate whose free variables are from $X \cup X'$. This predicate
    specifies when the mode switch using `$e$' is possible. It also
    specifies the updated values of the variables when this mode switch
    happens.
  \end{itemize}
  \label{def:ha}
\end{definition}

The semantics of an HA is specified using the notion of timed transitions
systems (TTS) using Definition~\ref{def:semanticsHA}.

\begin{definition}
  The semantics of an HA \newline
  $H = \langle Loc, Edge, \Sigma, Inv, Flow, Jump \rangle$ is
  defined by a timed transition system 
  $TTS =\langle Q, Q_0, \Sigma, \rightarrow \rangle$
  \begin{itemize}
  \item $Q$ is of the form $(l, v)$ where $l$ is a location and
    $v \in [X \rightarrow \mathbb{R}]$ such that $v$ satisfies
    $Inv(l)$. $Q$ is called the state-space of H.
  \item $Q_0 \subseteq Q$ of the form $(l, v)$ such that $v$ satisfies
    $Init(l)$.
  \item $\rightarrow$ is the set of transitions consisting of either:
    \begin{enumerate}
    \item Discrete transitions (instantaneous): For each edge \newline
      $e=(l,\sigma,l')$, we have
      $(l,v) \stackrel{\sigma} \rightarrow (l',v')$ if $(l,v) \in Q$,
      $(l',v') \in Q$ and $(v, v')$ satisfy $Jump(e)$.
    \item Continuous transition (delay): For each non-negative real
      $\delta$, we have $(l,v) \stackrel{\delta} \rightarrow (l,v')$ if
      $(l,v) \in Q$, $(l,v') \in Q$, and there is a differentiable
      function $f: [0,\delta] \rightarrow \mathbb{R}^{m}$ (which is a
      witness function of the above transition where $m=|X|$) with first
      derivative $\dot{f}: (0, \delta) \rightarrow \mathbb{R}^{m}$ such
      that the following conditions hold:
      \begin{itemize}
      \item $f(0) = v$ and $f(\delta)=v'$,
      \item for all $\epsilon \in (0, \delta), f(\epsilon)$
        satisfies $Inv(l)$ and $(f(\epsilon), \dot{f}(\epsilon))$
        satisfies $Flow(l)$.
      \end{itemize}
    \end{enumerate}
  \end{itemize}
  \label{def:semanticsHA} 
\end{definition}

The TTS semantics of an HA encodes the state-space $Q$ as all possible
pairs of the form $(l, v)$, where $l$ is a location and $v$ is a vector
representing the valuation of the variables. \ignore{For some $x \in X$ and for
some $t \in \mathbb{R}^{\geq 0}$, we write $x(t) \in v$, to mean the
valuation of the variable $x$ at time instant $t$.}  In general, the
state-space is infinite as there are infinite possible valuations of
variables within any time-step $\delta$.  Transitions may be any one
of two types. (1) Discrete transitions that lead to a mode switch and 
these transitions are  \emph{instantaneous}.  This
happens when the input event is present, the guard condition is true and
the updates are performed as the transition is taken. (2) Continuous
transitions that capture the passage of time when control resides in a
location.


\section{Semantic restrictions for efficient code generation}
\label{sec:synchronous}
The HA model introduced in the previous section is very generic,
expressive, and hence powerful.  However, it has limitations for code
generation. Since we want to generate code for emulating physical processes,
i.e., for PoC design, the efficient, real-time response of the generated code
is critical. \ignore{ While allowing arbitrary predicates to describe
  guards and invariants is powerful, code generation for arbitrary
  predicates means that we have to evaluate predicates for satisfiablity
  and validity, both of which are complex operations whose time
  complexity is NP-complete. Even more important is the fact that with
  arbitrary predicates and dependencies between HAs, modular code
  generation is infeasible.}

Another limitation of generic HA must also be carefully considered for
code generation.  As argued in the introduction, all known tools
(including ones widely used in industry) solve ODEs at
run-time using numeric techniques. Simulink\textsuperscript{\textregistered} also
introduces semantic ambiguities and there are problems with simulation
fidelity, due to the dependence on the step size selected for the
numerical solver~\cite{alur08symbolic}. For example, a numerical solver may fail to detect the peak of
a continuous variable and may perform backtracking and other heuristic
techniques, that affect the outcome of the simulation.  As a
consequence, robust tools, such as Ptolemy, are advocating the use of
techniques such as quantized-state systems~\cite{brooks15} as an alternative to dynamic
numerical simulation.

Considering the difficulty posed by the above issues, we propose the
well-formedness criteria for admitting a sub-class of HA, that we term 
\acf{WHA}, which is amenable to automatic code generation. Arbitrary predicates for describing invariants and
guard conditions are replaced by more structured conditions 
inspired by clock constraints in timed automata~\cite{alur94}. We also
move away from dynamic numerical solvers and instead use analytic solutions. Most 
importantly, we develop a synchronous semantics for the
generated code that provides key advantages for modular compilation and
sound transformations that are semantic preserving.

\subsection{Well-formed hybrid automata (WHA)}
\label{sec:wha}

A set is defined called $CV(X)$ that will be used for creating
constraints over the continuous variables. This is provided in
Definition~\ref{def:cvx}. These constraints are either comparison of a
variable with a natural number or Boolean combinations of these.
Comparisons can only be performed using the following mathematical operators:
$<, \leq, >, \geq$.

\begin{definition}
  Given a set of continuous variables $X$, then a set of constraints
  over these variables are defined as follows: \newline
  $g:=x<n | x \leq n| x > n | x \geq n | g \wedge g$ where,
  $n \in \mathbb{N}$ and $x \in X$. $CV(X)$ denotes the set of
  constraints over $X$. Without loss of generality we can extend these
  constraints to the domain of rational numbers $\mathbb{Q}$.
  \label{def:cvx} 
\end{definition}

A WHA must satisfy the following conditions and are defined in
Definition~\ref{def:wha}.

\begin{itemize}
\item All invariants, jump and init predicates are members of the set
  $CV(X)$.
\item In HA semantics, when control resides in a location, there are
  infinite valuations of variables in any given interval of time. This
  makes code generation difficult.  To facilitate code generation, we
  make evaluations only at discrete intervals. These intervals
  correspond to the \emph{ticks} of a synchronous program that will be
  used for code generation.
\item The ODEs which define flow constraints in any location of the form
  $\dot{x}=f(x)$ must be \emph{closed form} in nature. This property
  ensures that such ODEs are analytically solvable so that the witness
  functions needed for the generated code are analytically computable.
\item All witness functions must be monotonic. This necessary property ensures 
  that generated code can evaluate invariants
  and jump predicates using \emph{saturation} rather than backtracking.
   This is detailed in 
  Section~\ref{sec:saturation-function}.
\end{itemize}

We have defined WHA for both modular code generation and to enable
formal verification, which is not the focus of the current paper. If
only modular code generation is needed, some of the WHA restrictions may
be relaxed further.

\begin{definition}
  A well-formed hybrid automata \newline
  $WHA = \langle Loc, Edge, \Sigma, Init, Inv, Flow, Jump \rangle$ where
  \begin{itemize}
  \item $Loc=\{l_1,..,l_n\}$ representing $n$ control modes or
    locations.
  \item $\Sigma = \Sigma_I \cup \Sigma_O \cup \{\tau\}$ is the input
    alphabet comprising of event names, including internal events.
  \item $Edge \subseteq Loc \times \Sigma \times Loc$ are the set of
    edges between locations.
  \item Three sets for the set of continuous variables, their rate of
    change and their updated values represented as follows: \newline
    $X=\{x_1,.., x_m\}$, $\dot{X}=\{\dot{x_1},.., \dot{x_m}\}$, and \newline
    $X'=\{x_{1}',.., x_{m}'\}$.
  \item $Init(l)$: $x_1=v_1, .., x_m=v_m$ where
    $v_1..v_m \in \mathbb{R}$. 
  \item $Inv(l)$:  $\in CV(X)$
  \item $Flow(l)$: A set of ODEs that meet the closed form
    requirement. Also, the witness functions of the ODEs are
    monotonic.
  \item $Jump(e)$: This function maps each edge `$e$' to the conjunction
    of a guard and an update.  The guard is in $ CV(X)$ and the update
    is in $UP(X,X')$, which specifies the value of the updated
    variables.
  \item Fairness: the system never remains in any location indefinitely,
    i.e., if a location invariant holds indefinitely, then a egress
    transition is enabled by the controller within a bounded time.
  \end{itemize}
  \label{def:wha}
\end{definition}

The semantics of a WHA is also a TTS (Definition 2).

\subsection{Background on the synchronous approach}
\label{sec:synchronous-background}

The \emph{synchronous approach}~\cite{benveniste03} is 
widely used for the management of concurrency and race
conditions. It is, in particular, applied extensively to the design of
safety-critical systems, such as the embedded software for the Airbus A320
which uses the SCADE language and associated tools~\cite{scade}.

The synchronous approach is ideal for the development of \emph{reactive
systems}~\cite{harel:ReactiveSystems}.  A reactive system reacts to
its environment continuously and, in order to remain reactive, the speed
of the environment must be considered, i.e. outputs must be produced for
the current set of inputs before the next set of inputs appear.

The synchronous paradigm assumes that the \emph{idealised} reactive
system produces outputs synchronously relative to inputs, i.e., the
outputs are produced after zero delay. This is possible if the reactive
system executes infinitely fast relative to its environment, which is
known as the \emph{synchrony hypothesis}~\cite{benveniste03}.  Based on
this hypothesis, time may be treated as a sequence of discrete instants
or \emph{ticks} with nothing happening between the completion of the
current tick and the start of the next tick. This idea is prevalent in
various fields of engineering such as digital logic design and control
systems~\cite{ogata10}. Typically, the simulation step size of the ODE
solver used for the plant must match the sampling step size of the
controller for correct system simulation~\cite{carlsson2012methods} and
this aspect is often difficult to achieve. Unlike this, in the
synchronous approach, the notion of synchronous composition
~\cite{yoong2015model} formalises this automatically using compilation
technology.

The synchrony hypothesis is valid so long as the minimum inter-arrival
time of input events is longer than the maximum time required to
perform a reaction or tick.  This requires the computation of 
the worst-case reaction time (WCRT) of the synchronous
program~\cite{ju:PerformanceEsterel,roop:TightWCRT,Wang13ILP}.

Synchronous languages, which are based on the synchrony hypothesis,
include Esterel~\cite{berry2000foundations},
Lustre~\cite{halbwachs1991synchronous} and
Signal~\cite{leguernic1991programming}, where every program reaction
occurs with respect to a logical clock that \emph{tick}s.  A new
reaction starts at the beginning of a tick by taking a snapshot of the
input signals, computing the outputs using the user specified logic, and
emitting the output signals before the next tick starts. This is similar
to the scan cycle of a programmable logic controller~\cite{yoong2015model}.

\subsection{Synchronous HA (SHA)}
\label{sec:sha}

Because hybrid systems extend reactive systems (which are usually discrete)
with continuous dynamics, we need to make semantic adaptations to the
synchronous model to facilitate code generation.  We make the following
assumptions:

\begin{itemize}
\item All ODEs can be solved using analytic methods to compute their witness
  functions.
\item Execution is performed in discrete instants based
  on the synchronous approach and the duration of each instant is a
  fixed time $\delta$.
\item Due to the fairness assumption, the execution time
  spent in any location is always bounded.
\end{itemize}

\begin{figure}[ht]
  \centering
  \begin{tikzpicture}[transform shape, xscale=0.9]
\begin{axis}
[ xlabel={Time (in ticks)},
ylabel={Value of $x$},
axis y line = left,
axis x line = bottom,
extra y ticks={16,25},
extra tick style={grid=major}
]

\addplot[blue,  thick, smooth,
		domain=0:7] (x,x^2);

	 


\addplot[color=red,fill=red,only marks,mark=*] coordinates{(1,1)(2,4)(3,9)(4,16)(5,25)(6,36)(7,49)};

\node (WF) at (axis cs:2,40) {Witness Function}; 
\draw[->, thick] (WF) -- (axis cs:5.5,30.25);

\node (SA) at (axis cs:3.2,45)
	 {Synchronous Approximation}; 
\draw[->, thick] (SA) -- (axis cs:6,36);

\node (Vx4) at (axis cs:2.2,17)
	 {$x$ at tick 4};
\node (Vx5) at (axis cs:2.2,26)
	 {$x$ at tick 5};
	 
\end{axis}
\end{tikzpicture}
  \caption{An example visualisation of the synchronous approximation.
  	 Only observable at tick boundaries}
  \label{fig:sha-trace}
\end{figure}
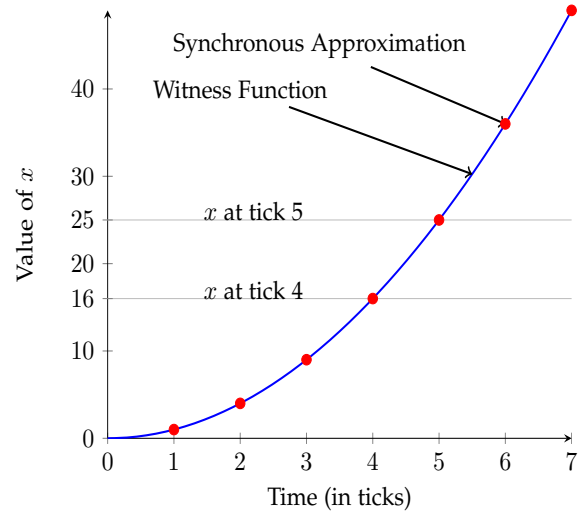

Based on these assumptions, Definition~\ref{def:sha} formalises the
concept of SHA corresponding to any WHA. A SHA provides an
\emph{under-approximation} of the behaviour of a WHA (a subset of
behaviours) due to the fact that in an SHA, the valuation of continuous
variables are made at discrete instants and the value remains constant
between two distinct valuations. This is visualised in
Figure~\ref{fig:sha-trace} and shows that the value of the
continuous variable $x$ is evaluated at discrete instants (or ticks)
and remains unchanged between two instants (or ticks).

We also assume the following notation. Let
$ValueInterval=\{[N_1, N_2]| N_1 \in \mathbb{N}~\wedge~N_2 \in
\mathbb{N}\}$.

\begin{definition}
  Given a well-formed hybrid automata \newline
  $WHA = \langle Loc, Edge, \Sigma, Init, Inv, Flow, Jump \rangle$ a SHA corresponding to
  this WHA is \newline
  $SHA = \langle Loc, Edge, \Sigma, Init, Inv, Switness, Jump, Step,$
  $Nsteps, BoundryCond \rangle$ where:
  \begin{itemize}
  \item $Step = \delta \in \mathbb{R}^+$: This specifies the duration of
    the synchronous instant.
  \item $Nsteps: Loc \rightarrow \mathbb{N}$: This specifies the maximum
    number of (logical) time steps that can be spent in any location. 
    $Nsteps(l) = k$ implies that
    during any execution of the SHA the time spent in
    $l$ must be less than or equal to $k \times \delta$.
  \item
    $Switness: Loc \times \mathbb{N} \times \mathbb{R} \times
    \mathbb{R}^m \rightarrow \mathbb{R}^m$:
    is the witness function that returns the valuation of all the
    continuous variables at any valid time step $k$ in a given location
    $l$. It also takes as input the time step $\delta$ and the initial
    value from which execution begins in the location.
  \item $BoundaryCond: (X \times Loc) \rightarrow ValueInterval$: This
    function defines the interval in which the boundary value of any
    continuous variable lies in a given location. The boundary value of
    a continuous variable $x$ is the value of the variable at time
    $0$, when execution starts in that location, i.e., $x(0)$.
  \end{itemize}
  \label{def:sha}
\end{definition}

A SHA is an abstraction of the corresponding WHA. It inherits all the
components of a WHA except that $Flow$ predicates are replaced by a
composite witness function $Switness$.
$Switness(l,k, \delta, i)=v_{k,l,i}$, which returns the evaluation of
the witness function in location $l$ at time step $k$. Here, $v_{k,l,i}$
is a vector representing the valuation of variables in $X$. We also
denote $v_{k,l,i}(x) \in \mathbb{R}$ to represent the valuation of the
continuous variable $x \in X$ in the $k$-th step in location $l$. We use
the shorthand $x[k, l, i]$ to denote $v_{k,l,i}(x)$ and the shorthand
$x[k]$ when the location $l$ and the initial value of $x$ in vector $i$,
itself denoted as $i(x)$ is clear from the context.

We also introduce some other new components:
$Step$, $Nsteps$, and $BoundaryCond$.  $Step$ specifies the duration of a
discrete instant or tick and $Nsteps$ maps every location to a natural
number indicating the worst-case number of steps possible in that
location during any execution of the SHA. $Step$ can be any value on the
positive real-number line, usually obtained via \textit{worst-case
  reaction time} (WCRT) analysis~\cite{ju:PerformanceEsterel,roop:TightWCRT,Wang13ILP}. 
 $Nsteps$ is computed statically
using an algorithm presented in
Section~\ref{sec:generation-sha}. $BoundaryCond(l,x)$ returns a closed
interval of the form $[N_1,N_2]$, which means that the boundary value of
$x(0)$ in location $l$ is in $[N_1,N_2]$. This mapping is
essential for computing the constants of integration and is explained in
Section~\ref{sec:generation-sha}.

The semantics of a \ac{SHA} is provided as a \ac{DTTS} in
Definition~\ref{def:dtts}. We assume that all transitions of a DTTS
trigger relative to the ticks of the logical clock of the synchronous
program.

\begin{definition}
  The semantics of a \newline
  $SHA = \langle Loc, Edge, \Sigma, Init, Inv, Switness, Jump, Step,$
  $Nsteps, BoundaryCond \rangle$ is a 
  $DTTS = \langle Q, Q^0, \Sigma, \rightarrow \rangle$ where

  \begin{itemize}
  \item The state-space is $Q$, where any state is of the form
    $(l, v, i, k)$ where $l$ is a location, $i$ is the initial valuation
    of the variables when execution begins in the location and $v$ is
    the valuation at the $k$-th instant.
  \item $Q^0 \subseteq Q$ where every $q^0 \in Q^0$ is of the form
    $(l, v^0, i, k)$ such that $v^0$ satisfies $Init(l)$.
  \item Transitions are of two types:
    \begin{itemize}
    \item \emph{Inter-location transitions} that lead to mode switches:
      These are of the form
      $(l, v, i, k) \stackrel{\sigma} \rightarrow (l', v', i', 0)$ if
      $(l, v, i, k) \in Q$, $(l', v', i', 0) \in Q$,
      $e=(l \stackrel{\sigma} \rightarrow l') \in Edge$ and $(v, v')$
      satisfy $Jump(e)$.
    \item \emph{Intra-location transitions} made during the execution in
      a given mode / location: These are of the form
      $(l, v, i, k) \rightarrow (l, v', i, k+1)$ if
      $(l, v, i, k) \in Q$, $(l, v', i, k+1) \in Q$, $(v, v')$ satisfy
      $Inv(l)$, $Switness(l,k,\delta,i)=v$ and $Switness(l,k+1,\delta,i)=v'$.
    \end{itemize}
  \end{itemize}
  \label{def:dtts}
\end{definition}

\subsection{Composition}
\label{sec:DTTS-par}

The developed synchronous semantics enables modular code generation. We
can compile each WHA separately (step 3) in Figure~\ref{fig:methodology}
and then the generated codes are linked together in step 4, which is
similar to the linking process in conventional compilation. We formalise
the synchronous parallel composition of SHAs using
Definition~\ref{def:parallel}. This facilitates the seamless linking
process, described in Section~\ref{sec:composition-sha}.


\begin{definition}
  \label{sec:sha-semant-comp}
  Given \newline 
  $SHA_1 = \langle Loc_1, Edge_1, \Sigma_1, Init_1, Inv_1, Switness_1, Jump_1,$ $Step_1, Nsteps_1, BoundaryCond_1 \rangle$  and its semantics \newline
  $DTTS_1 = \langle Q_1, Q^0_1, \Sigma, \rightarrow_1 \rangle$ and \newline
  $SHA_2 = \langle Loc_2, Edge_2, \Sigma_2, Init_2, Inv_2, Switness_2, Jump_2,$ $Step_2, Nsteps_2, BoundaryCond_2 \rangle$ and its semantics \newline
  $DTTS_2 = \langle Q_2, Q^0_2, \Sigma_2, \rightarrow_2 \rangle$ 
  and $Shared_v=X_1 \bigcap X_2$
  and $Shared_e=\Sigma_1 \bigcap \Sigma_2$ denoting a set of shared
  variables and events, respectively, which satisfy the following
  conditions:
      
  \begin{enumerate}
  \item All writes to shared variables $x \in Shared_{v}$ or the
    emission of shared events $e \in Shared_e$ are mutually exclusive.
  \item Any read to a shared variable $x \in Shared_{v}$ or a shared
    event $e \in Shared_e$ accesses the value of the variable / event in
    the previous tick, denoted $pre(x)/pre(e)$.
  \end{enumerate}

  $DTTS_1 || DTTS_2 = DTTS \langle Q, Q^0, \Sigma, \rightarrow \rangle$ where:
  \begin{itemize}
  \item The state-space is $Q \subseteq Q_1 \times Q_2$.
  \item $Q^0 \subseteq Q^0_1 \times Q^0_2$.
  \item Transitions $\rightarrow$ are of two types:
    \begin{itemize}
    \item \emph{Inter-location transitions} of the form:\\
      \begin{flushleft}
      \noindent (\SHArule{Inter-Inter})
      $(q_1, q_2) \stackrel{\sigma_1 \wedge \sigma_2} \rightarrow (q_1', q_2')$ where 
      $q_1=(l_1, v_1, i_1, k)$, $q_2=(l_2, v_2, i_2, k)$,
      $q_1'=(l_1', v_1', i_1', 0)$, $q_2'=(l_2', v_2', i_2', 0)$.
      $e_1 = (q_1 \stackrel{\sigma_1} \rightarrow q_1') \in Edge_1$ and
      $(v_1, v_1')$ satisfy $Jump(e_1)$.
      $e_2 = (q_2 \stackrel{\sigma_2} \rightarrow q_2') \in Edge_2$ and
      $(v_2, v_2')$ satisfy $Jump(e_2)$. \newline

      \noindent (\SHArule{Inter-Intra})
      $(q_1, q_2) \stackrel{\sigma_1} \rightarrow (q_1', q_2')$ where
      $q_1=(l_1, v_1, i_1, k)$, $q_2=(l_2, v_2, i_2, k)$,
      $q_1'=(l_1', v_1', i_1', 0)$, $q_2'=(l_2, v_2', i_2', 0)$.
      $e_1 = (q_1 \stackrel{\sigma_1} \rightarrow q_1') \in Edge_1$ and
      $(v_1, v_1')$ satisfy $Jump(e_1)$. $q_2, q_2' \in Q_2$ and $(v_2,
      v_2')$ satisfy $Inv(l_2)$. Finally, $i_2' = v_2$ and $v_2' = v_2$. \newline
        
      \noindent (\SHArule{Intra-Inter})
      $(q_1, q_2) \stackrel{\sigma_2} \rightarrow (q_1', q_2')$ where
      $q_1=(l_1, v_1, i_1, k)$, $q_2=(l_2, v_2, i_2, k)$,
      $q_1'=(l_1, v_1', i_1', 0)$, $q_2'=(l_2', v_2', i_2', 0)$.
      $e_2 = (q_2 \stackrel{\sigma_2} \rightarrow q_2') \in Edge_2$ and
      $(v_2, v_2')$ satisfy $Jump(e_2)$. $q_1, q_1' \in Q_1$ and $(v_1,
      v_1')$ satisfy $Inv(l_1)$. Finally, $i_1' = v_1$ and $v_1' = v_1$. \newline
      \end{flushleft}
      
    \item \emph{Intra-location transitions} of the form:\\
      \begin{flushleft}
      \noindent (\SHArule{Intra-Intra}) $(q_1, q_2) \rightarrow (q_1', q_2')$
      where \newline 
      $q_1=(l_1, v_1, i_1, k)$, $q_2=(l_2, v_2, i_2, k)$,
      $q_1'=(l_1, v_1', i_1, k+1)$, $q_2'=(l_2, v_2', i_2, k+1)$, and
      $q_1, q_1' \in Q_1$, $(v_1, v_1')$ satisfy $Inv(l_1)$,
      $Switness(l_1,k,\delta, i_1)=v_1$ and $Switness(l_1,k+1,\delta, i_1)=v_1'$. Similarly,
      $q_2, q_2' \in Q_2$, $(v_2, v_2')$ satisfy $Inv(l_2)$,
      $Switness(l_2,k,\delta, i_2)=v_2$ and $Switness(l_2,k+1,\delta, i_2)=v_2'$
      for any $\delta$.
      \end{flushleft}
    \end{itemize}
  \end{itemize}
  \label{def:parallel}
\end{definition}

In Definition~\ref{def:parallel}, the DTTS corresponding to two SHAs
composed in parallel is computed by composing their respective
DTTSs. The state-space $Q$ of the resultant DTTS is a subset of the
product of the state-space of the individual constituents. The initial
state-space $Q^0$ is also a subset of the product of the initial state-space 
of the constituents. The transition relation consists of two types
of transitions.  The inter-location transitions happen when any one of
the constituents or both constituents make an inter-location transition
(there are three different possibilities). \SHArule{Inter-Inter} states that both constituents can
take a discrete transition to new locations. 

\SHArule{Inter-Intra} states that when
the first constituent $Q_1$ takes an inter-location transition from
location $l_1$ to $l_1'$, $Q_2$ is also \textit{forced} to make such a
transition.  But, the resultant location of $Q_2$ does not change, i.e.,
$Q_2$ can only take a transition from $l_2$ to $l_2$.  Moreover, upon
taking the transition the initial value  $i_2$ is set to the
current valuation $v_2$, consequently implying that $v_2' = v_2$. Once
in the new state $(l_1', l_2)$ vectors $v_1$ and $v_2$ start evolving
according to their individual witness functions. \SHArule{Intra-Inter} is the
dual of \SHArule{Inter-Intra}. 
Finally, intra-location transition in the composition
(\SHArule{Intra-Intra}) happens only when both constituents make an intra-location
transition.


\section{Methodology for code generation}
\label{sec:code-gen}
\begin{figure}[t!]
  \centering
  \scalebox{0.7}{
\tikzstyle{decision} = [diamond, draw, fill=white!30, 
    text width=5em, text badly centered, node distance=3.2cm, inner sep=0pt]
\tikzstyle{file} = [rectangle, draw, fill=gray!20, 
    text width=5em, text centered, minimum height=4em]
\tikzstyle{process} = [rectangle, draw, fill=blue!15, 
    text width=5em, text centered, rounded corners, minimum height=4em]
\tikzstyle{line} = [draw, -latex', ultra thick]
\tikzstyle{state} = [draw, ellipse,fill=red!30, 
	node distance=3.5cm, text width=5em, 
	text badly centered,
    minimum height=2em]

\begin{tikzpicture}[node distance = 2cm, auto]
    \node [file] (HA) {{HA}};

    \node [decision, right of = HA] (D) {Is a WHA~? (step 1,
      Sec~\ref{sec:static-analysis-ha})};

    \node [state, right of = D] (E) {Invalid input};

    \node [file, right of = E, node distance=3cm] (FSM) {FSM/ C-code};

    \node [process, below of = D, node distance =3cm] (GSHA) {Generate
      \ac{SHA} (step 2, Sec~\ref{sec:generation-sha})};
    
    \node [file, below of = E, node distance =3cm] 
	    (SHA) {\ac{SHA}};
    
            \node [process, below of = FSM, node distance =3cm] (GFSM)
            {Generate backend code (step 3,
              Sec~\ref{sec:back-code-gener})};

	\path [line] (HA) -- (D);
	\path [line] (D)-- node[near start]{no}(E);
	\path [line] (D)-- node[near start]{yes}(GSHA);
	
	\path [line] (GSHA) -- (SHA);
	\path [line] (SHA) -- (GFSM);
	\path [line] (GFSM) -- (FSM);

\end{tikzpicture}    
  }
  \caption{Process for compiling a single \acf{HA}}
  \label{fig:1}
\end{figure}
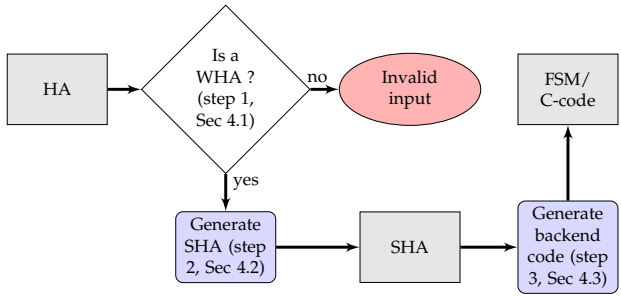

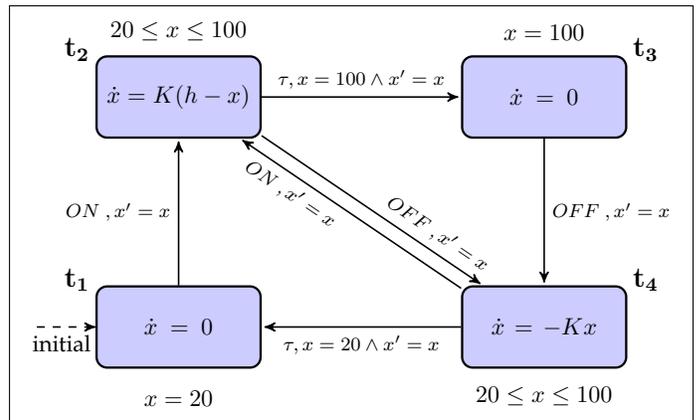
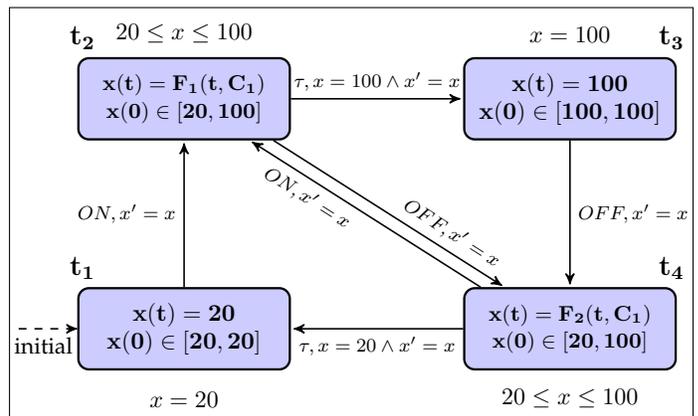
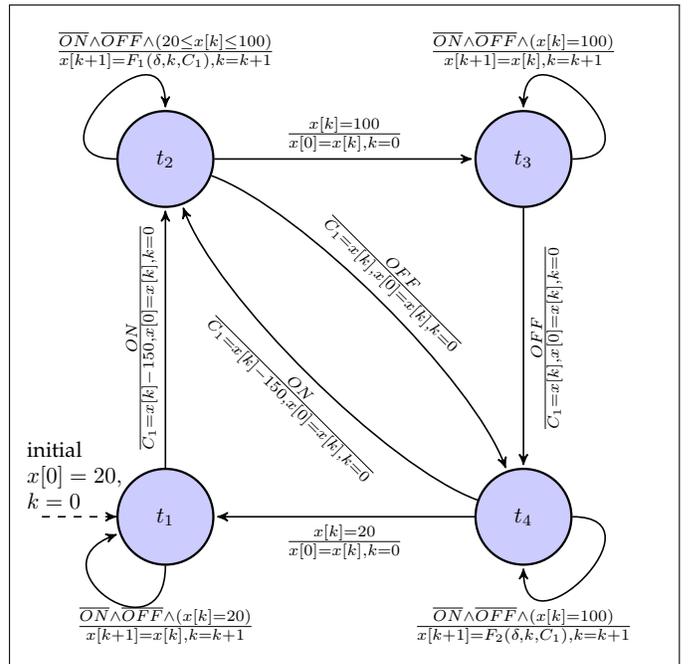
\begin{figure}[htbp]
  \centering
  {
	\centering
	\subfigure[\acf{HA}, see Definition~\ref{def:ha}. Reproduced from Figure~\ref{fig:waterTankHAtank}. Symbols $K$ and $h$ are constants with values $0.075$ and $150$, respectively. \label{fig:codeGenHA}
	]{
	\framebox[0.50\textwidth]{
		\begin{tikzpicture}[->,>=stealth',shorten >=1pt,auto,
			node distance=5.4cm,
			semithick,scale=0.9, transform shape]
			\tikzstyle{every state}=[rectangle,rounded corners,
			minimum height = 1.2cm, text width=2.2cm, text centered, fill=blue!20,draw=none,text=black, draw,line width=0.3mm]
			
			\node[state, 
			label={[shift={(0,-1.9)}]$ x = 20$}, 
			label={[shift={(-1.5,-0.2)}]\large $ \mathbf{t_1}$ }]
			(T1)  {$\dot{x}=0$};
			
			\draw[<-, dashed](T1.180) -- node[below]
                        {initial} ++(-1.0cm,-0cm);
			
			\node[state, 
			label={[shift={(0,0.1)}]$20\leq x \leq 100$}, 
			label={[shift={(-1.5,-0.2)}]\large $\mathbf{t_2}$   }]
			(T2) [node distance=3.4cm, above of=T1] 	 {$\dot{x}=K(h-x)$};
			
			\node[state, 
			label={[shift={(0,0.1)}]$ x = 100$}, 
			label={[shift={(1.5,-0.2)}]\large $\mathbf{t_3}$  }] 
			(T3) [right of=T2] 	 {$\dot{x}=0$};
			
			\node[state, 
			label={[shift={(0,-1.9)}]$20\leq x \leq 100$}, 
			label={[shift={(1.5,-0.2)}]\large $\mathbf{t_4}$  }]
			(T4) [right of=T1] 	 {$\dot{x}=-Kx$};

			\draw (T1) -- (T2) node [midway] 
			{\footnotesize $\inSignal{ON},x^\prime = x$};
			\draw (T2) -- (T3) node [midway] 
			{\footnotesize $\iSignal,x=100 \wedge x^\prime = x$};
			\draw (T3) -- (T4) node [midway, right] 
			{\footnotesize $\inSignal{OFF},x^\prime = x$};
			\draw (T4) -- (T1) node [midway] 
			{\footnotesize $\iSignal,x=20 \wedge x^\prime = x$};
			
			\draw (T2.335) -- (T4) node [midway, sloped] 
			{\footnotesize $\inSignal{OFF},x^\prime = x$};
			\draw (T4.155) -- (T2) node [midway, sloped] 
			{\footnotesize $\inSignal{ON},x^\prime = x$};
		\end{tikzpicture}
	  } 
	} 
	
	\subfigure[\label{fig:codeGenSHA} \acf{SHA}, see
        Definition~\ref{def:sha}. Flow predicates are described using
        witness functions. Values of the constants $h$ and $K$ are $150$
        and $0.075$, respectively. Furthermore, witness function
        $F_1=C_1 e^{-0.075\times t} + 150.0$ and
        $F_2=C_1 e^{-0.075\times t}$ ]{
          \framebox[0.5\textwidth]{
		\begin{tikzpicture}[->,>=stealth',shorten >=1pt,auto,
		node distance=5.7cm,
		semithick,scale=0.9, transform shape]
		\tikzstyle{every state}=[rectangle,rounded corners,
		minimum height = 1.2cm, text width=2.9cm, text centered, fill=blue!20,draw=none,text=black, draw,line width=0.3mm]
		
		\node[state, 
		label={[shift={(0,-1.9)}]$ x = 20$}, 
		label={[shift={(-1.5,0)}]\large $ \mathbf{t_1}$ }]
		(T1)  { $\mathbf{x(t)=20}$ $\mathbf{x(0) \in [20,20]}$};

		\draw[<-, dashed](T1.180) -- node[below] {initial}
                ++(-1cm,0cm);

		\node[state, 
		label={[shift={(0,0.1)}]$20\leq x \leq 100$}, 
		label={[shift={(-1.5,0)}]\large $\mathbf{t_2}$   }]
		(T2) [node distance=3.4cm, above of=T1]  	 {\small  $\mathbf{x(t)=F_1(t,C_1)}$  $\mathbf{x(0) \in [20,100]}$};
		
		\node[state, 
		label={[shift={(0,0.1)}]$ x = 100$}, 
		label={[shift={(1.5,0)}]\large $\mathbf{t_3}$  }] 
		(T3) [right of=T2] 	 {$\mathbf{x(t)=100}$  $\mathbf{x(0) \in [100,100]}$};
		
		\node[state, 
		label={[shift={(0,-1.9)}]$20\leq x \leq 100$}, 
		label={[shift={(1.5,0)}]\large $\mathbf{t_4}$  }]
		(T4) [right of=T1] 	 {\small  $\mathbf{x(t)=F_2(t,C_1)}$  
			$\mathbf{x(0) \in [20,100]}$};

		\draw (T1) -- (T2) node [midway] 
		{\footnotesize $ON,x^\prime = x$};
		\draw (T2) -- (T3) node [midway] 
		{\footnotesize $\iSignal,x=100 \wedge x^\prime = x$};
		\draw (T3) -- (T4) node [midway, right] 
		{\footnotesize $OFF,x^\prime = x$};
		\draw (T4) -- (T1) node [midway] 
		{\footnotesize $\iSignal,x=20 \wedge x^\prime = x$};

		\draw (T2.335) -- (T4) node [midway, sloped] 
		{\footnotesize $OFF,x^\prime = x$};
		\draw (T4.155) -- (T2) node [midway, sloped] 
		{\footnotesize $ON,x^\prime = x$};
		
		\end{tikzpicture}
		} 
	}
	
	\subfigure[\label{fig:codeGenDTTS} {Synchronous Witness Automata
          (SWA). We abuse the notation $x[k]$ to update the value of
          $x$, although $x[k]$ represents the valuation of the
          continuous variable $x$. This physical time
          $t=k \times \delta$, where $k$ is the logical tick and
          $\delta$ is the tick length.}]{ \framebox[0.49\textwidth]{
		\begin{tikzpicture}[->,>=stealth',shorten >=1pt,auto,node distance=5.6cm,
			semithick,scale=0.85, transform shape]
			\tikzstyle{every state}=[circle,rounded corners,
                        minimum height = 1.2cm, text width=1.2cm, text
                        centered, fill=blue!20,draw=none,text=black,
                        draw,line width=0.3mm]
			
			\node[state] (T1)  { $t_1$ };
			
			\draw[<-, dashed](T1.180) -- node[above, text
                        width =1.5cm] {initial $x[0]=20,$ $k=0$}
                        ++(-1.3cm,0cm);
			
			\node[state]
			(T2) [above of=T1] 	 {$t_2$};
			
			\node[state] 
			(T3) [right of=T2] 	 {$t_3$};
			
			\node[state] 
			(T4) [right of=T1] 	 { $t_4$};

			\draw (T1) -- (T2) node [pos=0.5, rotate =90, above] 
			{ $\frac{{ON}}{C_1=x[k]-150, x[0]=x[k],k=0}$};
			
			\draw (T2) -- (T3) node [midway] 
			{ $\frac{x[k]=100}{x[0]=x[k], k=0}$};
			
			\draw (T4) -- (T1) node [midway] 
			{ $\frac{x[k]=20}
				{x[0]=x[k],k=0}
				$};
			
			\draw (T3) -- (T4) node [midway,rotate =90, below] 
			{$\frac{ {OFF}}
				{C_1=x[k], x[0]=x[k],k=0} $};
			
			
			\draw (T2.180) to[out=180,in=90, distance=1.5cm] (T2.90);
			\draw (T2)  node[yshift=1.7cm]
			{ $\frac{\overline{ON}  \wedge \overline{OFF} \wedge (20 \leq x[k]  \leq 100)} {x[k+1]=F_1(\delta,k,C_1), k=k+1}   $ };
			
			\draw (T3.0) to[out=0,in=90, distance=1.5cm] (T3.90);
			\draw (T3)  node[yshift=1.7cm]
			{ $\frac{\overline{ON}  \wedge \overline{OFF} \wedge ( x[k]  = 100)} {x[k+1]=x[k],k=k+1}   $ };
			
			\draw (T4.0) to[out=0,in=-90, distance=1.5cm] (T4.-90);
			\draw (T4)  node[yshift=-1.7cm]
			{ $\frac{\overline{ON}  \wedge \overline{OFF} \wedge ( x[k]  = 100)}{x[k+1]=F_2(\delta,k,C_1),k=k+1}   $ };
			
			\draw (T1.-90) to[out=-90,in=200, distance=1.5cm] (T1.200);
			\draw (T1)  node[yshift=-1.7cm]
			{ $\frac{\overline{ON}  \wedge \overline{OFF} \wedge ( x[k]  = 20)}{x[k+1]=x[k],k=k+1}   $ };
			
			\draw (T2) to[out=-20,in=110, distance=1.5cm] (T4);
			\draw (T4)  node[yshift=3.7cm, xshift=-2cm, rotate =-45]
			{ $\frac{{OFF} }{C_1=x[k], x[0]=x[k],k=0}   $ };
			
			\draw (T4) to[out=160,in=-70, distance=1.5cm] (T2);
			\draw (T2)  node[yshift=-3.7cm, xshift=2cm, rotate =-45]
			{ $\frac{{ON}  }{C_1=x[k]-150, x[0]=x[k],k=0}   $ };
			
		\end{tikzpicture}
	  } 
	}

}

  \caption{The water tank component from the running example}
  \label{fig:codeGen}
\end{figure}

Figure~\ref{fig:1} outlines the approach used to compile a single HA.
The compilation process consists of three steps. We describe all three
steps, in
Sections~\ref{sec:static-analysis-ha}-\ref{sec:back-code-gener}, using
 the water tank HA, which has been
reproduced in Figure~\ref{fig:codeGenHA}.

\subsection{Static analysis of hybrid automata}
\label{sec:static-analysis-ha}

Given an HA, the first step (Figure~\ref{fig:1}) in the compilation
process is to statically determine if the well-formedness criteria
defined in Section~\ref{sec:wha} is respected. If the HA does not
respect the well-formedness criteria, then an error is generated.

\renewcommand{\algorithmiccomment}[1]{// #1}
\renewcommand{\algorithmicrequire}{\textbf{Input:}}
\renewcommand{\algorithmicensure}{\textbf{Output:}}
\begin{algorithm}
  \begin{algorithmic}[1]
    \REQUIRE HA ha 
    \ENSURE Boolean
    \STATE $gset \leftarrow \emptyset$ 
    \FORALL{$edges \in ha$} \label{alg:1:2}
    \FORALL{$guards \in edges$}
    \STATE $gset \leftarrow gset \cup guards$
    \ENDFOR
    \ENDFOR
    \FORALL{$loc \in ha$}
    \FORALL{$invs \in loc$}
    \STATE $gset \leftarrow gset \cup invs$
    \ENDFOR
    \ENDFOR
    \STATE \COMMENT{ \small check that all location invariants and jump conditions are    	of type $CV(X)$}
    \FORALL{$g \in gset$ }  
    \STATE $assert\ type(g) \in CV(X)$ \label{alg:1:3}\\
    \ENDFOR \label{alg:1:31}
    \FORALL{$loc \in ha$} \label{alg:1:4}
    \FORALL{$ode \in loc$}
    \IF{$solve\_ode(ode)$}\label{alg:1:6}
    \STATE $(R_1, R_2) \leftarrow solve(ode.rhs > 0)$\label{alg:1:6a} \footnotemark
    \STATE $(R'_1, R'_2) \leftarrow solve(ode.rhs < 0)$\label{alg:1:6b}
    \STATE $[N_1, N_2] \leftarrow get\_inv\_bounds(ode,loc)$\label{alg:1:6c}
    \IF{$(R_1,R_2) \cap [N_1, N_2] \neq \emptyset \wedge 
      (R'_1,R'_2) \cap [N_1, N_2] \neq \emptyset$}
    \STATE \COMMENT{ \small Slope of witness function changes sign}
    \STATE throw Exception (``Not a WHA'') \label{alg:1:7}
    \ELSE 
    \RETURN True
    \ENDIF
    \ELSE
    \STATE \COMMENT{\small No closed form solution, hence not a WHA}
    \STATE throw Exception(``Not a WHA'') \label{alg:1:8}
    \ENDIF
    \ENDFOR
    \ENDFOR\label{alg:1:5}
  \end{algorithmic}
  \caption{The algorithm to check the well-formedness criteria of a HA}
  \label{alg:1}
\end{algorithm}

The procedure to verify the well-formedness criteria is presented in
Algorithm~\ref{alg:1}, which takes an HA as input. Three well-formedness
criteria need to be guaranteed. First, the invariants and the
jump conditions need to be of the form $CV(X)$
(Definition~\ref{def:cvx}).  Lines~\ref{alg:1:2}-\ref{alg:1:31}
guarantee that this criterion is met. The second and third criteria
require that each ODE, in every location, of the HA should have a closed
form solution and should be monotonic. Lines~\ref{alg:1:4}-\ref{alg:1:5}
ensure that these criteria are satisfied.

Consider the running example HA -- the water tank system presented in
Figure~\ref{fig:codeGenHA}. Lines~\ref{alg:1:2}-\ref{alg:1:31} in
Algorithm~\ref{alg:1} collect all the invariant and jump conditions from
the locations and the edges, respectively. Once collected in set $gset$,
an assertion statement guarantees that all these conditions are of the
form $CV(X)$ (line~\ref{alg:1:3}).  Lines~\ref{alg:1:4}-\ref{alg:1:5}
iterate through each location of the HA. Upon visiting a location, all
ODEs within the location are solved symbolically
(line~\ref{alg:1:6}). If no closed form solution exists, then an exception
is generated (line~\ref{alg:1:8}).

Given that a closed form solution exists, we then guarantee that all ODEs
in a location are monotonic (not necessarily strictly monotonic). We use
the definition that any given (witness) function is considered monotonic
if and only if the first derivative of the function does not change
sign~\cite{rudin1987real}. Flow conditions evolve one or more ODEs
within a given location as long as the invariant on the location is not
violated. Hence, in our case, the definition of a monotonic function can
be made more specific: any given (witness) function is
monotonic if and only if its first derivative does not change sign within
the interval specified by the invariant(s) of the location.

The right hand side of the ODEs specify the first derivatives of the
witness functions. We need to ensure that the right hand side expression
of the ODE (the slope) does not change signs within the invariant
bounds. The lines~\ref{alg:1:6a}-\ref{alg:1:7} ensure that these
conditions are satisfied. Line~\ref{alg:1:6a} obtains the real number line
interval (denoted by $(R_1, R_2)$) such that the derivative of the
witness function is always greater than zero, i.e., an increasing
function. Line~\ref{alg:1:6b} obtains the real number line interval
(denoted by $(R'_1, R'_2)$), such that the first derivative of the witness
function is less than zero. Next, we obtain the invariant interval
(denoted $[N_1, N_2]$), bounding the value of the evolving variable in
the ODE, from the invariant(s) on the location. A non-empty interval
$(R_1,R_2)\cap [N_1, N_2]$ indicates that the witness function is
increasing within the location intervals. Similarly, a non-empty
interval $(R'_1, R'_2)\cap[N_1, N_2]$ indicates that the witness
function is a decreasing function within the location invariants. If
both sets are non-empty, the witness function increases and decreases
within the invariants specified on the location, and hence, the witness
function is not monotonic. 

Consider the running example in Figure~\ref{fig:codeGenHA} and
specifically consider location $t_2$. The steps to determine if the
witness function in location $t_2$ is monotonic are as follows: 
\begin{enumerate}
\item \textit{Compute the real number line interval for $x$ such that
    the right hand side of the ODE is strictly greater than zero}:
  \mbox{$x \in solve((0.075*(150-x)>0)$} results in
  \mbox{$x \in (-\infty, 150)$}
\item \textit{Compute the real number line interval for $x$ such that
    the right hand side of the ODE is strictly less than zero}:
  \mbox{$x \in solve((0.075*(150-x)<0)$} results in
  \mbox{$x \in (150, \infty)$}
\item \textit{Obtain the invariant bounds on $x$}: For location $t_2$,
  this interval is obtained from invariant: $20 \leq x \leq 100$. Hence,
  the invariant bounds are: $x \in [20, 100]$.
\item \textit{Check if the witness function is monotonic}: The
  intersection \mbox{$(-\infty, 150) \cap [20, 100] = [20, 100]$}, but
  \mbox{$(150, \infty) \cap [20, 100] = \emptyset$}, and hence, the
  witness function in location $t_2$ is an increasing (and monotonic)
  function.
\end{enumerate}

\footnotetext{$solve$ is a polymorphic function from Sympy (a python
  library) that can be used to solve many different types of
  problems. We make use of the polymorphic nature of $solve$ in
  Algorithms~\ref{alg:1} and \ref{alg:2}}

For the running example in Figure~\ref{fig:codeGenHA}, all the
well-formedness criteria are satisfied and hence, the water tank system
is a WHA.

\subsection{Generation of SHA}
\label{sec:generation-sha}

This section describes step 2 in the compilation procedure from
Figure~\ref{fig:1}. Once a given HA is is determined to be a WHA,
an SHA is generated from the WHA. This procedure translates all ODEs in
every location of the WHA into their closed form solutions (witness
functions). It also computes the worst-case bound
$Nsteps$ (if one exists) as defined in Definition~\ref{def:sha}.

\begin{algorithm}[t!]
  \begin{algorithmic}[1]
    \REQUIRE HA ha 
    \ENSURE SHA sha
    \FORALL{$loc \in ha$} 
    \FORALL{$ode \in loc$}
    \STATE $times \leftarrow \emptyset$
    \STATE $eq \leftarrow solve\_ode(ode)$ \label{alg:2:1}
    \STATE $ode \leftarrow eq$ \label{alg:2:8}
    \IF {$sign(\textit{diff}(eq, t)) > 0$} \label{alg:2:2}
    \STATE \mbox{$times \leftarrow times \cup
      solve(eq,min(Init),max(Inv))$}\label{alg:2:4}
    \ELSIF {$sign(\textit{diff}(eq, t)) < 0$}
    \STATE \mbox{$times \leftarrow times \cup
      solve(eq,max(Init),min(Inv))$}\label{alg:2:5}
    \ELSE
    \STATE raise Warning (``Fairness required'') \label{alg:2:6}
    \STATE $times \leftarrow \emptyset$
    \STATE break
    \ENDIF
    \ENDFOR
    \IF {$min(times) = \infty$}
    \STATE raise Warning (``Fairness required'') \label{alg:2:7}
    \ENDIF \label{alg:2:3}
    \ENDFOR
    \RETURN this
  \end{algorithmic}
  \caption{The algorithm to generate a SHA from a WHA}
  \label{alg:2}
\end{algorithm}

The procedure to generate the SHA from a WHA is presented in
Algorithm~\ref{alg:2}. The algorithm visits every location in the
WHA. For each ODE in the location, a witness function is obtained
(line~\ref{alg:2:1}). Next, the algorithm attempts to determine the 
maximum time that
the HA will spend in every location (lines~\ref{alg:2:2}-\ref{alg:2:3})
in the worst-case. The algorithm first detects the slope of each ODE
($ode$ in Algorithm~\ref{alg:2}) for any given location ($loc$ in
Algorithm~\ref{alg:2}), by differentiating the witness function with
respect to time $t$ (line~\ref{alg:2:2}). If the slope is increasing
(line~\ref{alg:2:2}) then the witness function ($eq$ in
Algorithm~\ref{alg:2}) is solved for $t$. We do the reverse for an ODE
with decreasing slope (line~\ref{alg:2:5}). If the slope is constant, in
case of a non-changing ODE, a warning is raised stating that $egress$
(fairness) condition is needed to make progress out of the location
(line~\ref{alg:2:6}). Given the various worst-case times in the set
$times$, for all ODEs in a location, the minimum value amongst $times$
needs to be bounded to guarantee at compile time that $Nsteps$ is
bounded, otherwise an egress condition is needed to make progress out of
the location (line~\ref{alg:2:7}).

Consider the example of the SHA shown in
Figure~\ref{fig:codeGenSHA}. The algorithm visits each location in the
WHA (Figure~\ref{fig:codeGenHA}) and replaces each ODE with its
equivalent closed form solution (Algorithm~\ref{alg:2},
line~\ref{alg:2:8}). The closed form solution of all ODEs for the water
tank is shown in Figure~\ref{fig:codeGenSHA}, given
$h=150$ and $K=0.075$, respectively.

Consider location $t_1$ in
Figure~\ref{fig:codeGenSHA}. Line~\ref{alg:2:2}, in
Algorithm~\ref{alg:2}, differentiates the witness function with respect
to time ($t$), resulting in a sign of $0$, indicating a non-changing
witness function. Hence, a warning that an $egress$ transition is needed
to make progress out of location $t_1$ is raised. This is expected,
because as seen in Figure~\ref{fig:codeGenHA}, the ODE $\dot{x} = 0$
does not change the value of $x$ and the invariant $x = 20$ holds, hence
progress out of location $t_1$ is only possible upon reception of
($egress$) event \verb$ON$.

Similarly, differentiating the witness function with respect to time in
location $t_4$, gives the solution $-0.075 \times C1$, where $C1$ is the
constant of integration. In our framework, the programmer annotates
every location with the possible range of (initial) values that the
continuous variable(s) might take upon entering the location. For
example, $BoundaryCond(t_4,x)$ is denoted as $x(0) \in [20, 100]$ for
location $t_4$\footnote{The possible initial value intervals can be
  computed automatically from the HA, but describing this is out of the
  scope of this paper.}. Given that the initial value interval is
positive, the slope of the witness function is detected to be negative
(since C1 can only take a positive value) and hence, the witness
function in location $t_4$ is a decreasing function. Once this is
deduced, the worst-case time spent in location $t_4$ is computed at
line~\ref{alg:2:5}, in Algorithm~\ref{alg:2}, by substituting the
initial value of $x$ as the maximum value from the
\textit{ValueInternal}, i.e., $x(0) = max(20, 100) = 100$ and the final
value of $x$ as the minimum value from amongst the invariant on $x$ in
location $t_4$, i.e., $x(t) = min(20, 100) = 20$.

Given that the witness function $x(t)$ is a decreasing function with an
initial value of $100$, we compute, using $solve$ on line~\ref{alg:2:5},
the time $t$ it takes for function $x(t)$ to reach its final value
$20$. This time $t$, if it can be computed, gives the maximum time that
the system will remain in location $t_4$ before making progress to a new
location. The \textit{solve} procedure proceeds as follows:

\begin{enumerate}
\item \textit{Compute the value of the constant of integration}: For
  location $t_4$, $x(t) = C_1 \times e^{-0.075\times t}$ as shown in
  Figure~\ref{fig:codeGenSHA}. First of all, we compute the constant of
  integration, using the initial value
  $x(0) = 100 \Rightarrow 100 = C_1 \times e^{-0.075 \times 0}$.
  Therefore, $C_1 = 100$ and hence $x(t) = 100 \times e^{-0.075 \times
    t}$.

\item \textit{Computing time $t$}: Once, we have computed the constant
  of integration, we can easily compute the time $t$ 
  required for $x(t)$
  to decrease from $100$ to $20$
  as follows:
  $x(t) = 20 \Rightarrow 20$. Therefore, $20 = 100 \times e^{-0.075 \times t}$. Hence
  $t = \frac{ln(\frac{20}{100})}{-0.075}, t \approx 214.6$. 
\end{enumerate}

The system will make progress out of location $t_4$, in the worst-case,
after approximately, $214.6$ units of time. Similarly, we can bound the
worst-case time spent in a location with increasing, rather than
decreasing witness functions. In case a location consists of multiple
witness functions, the worst-case time spent in that location is the
minimum amongst the worst-case times computed for all the witness
functions in that location. If the minimum amongst all the worst-case
times is $\infty$, then a warning stating that an egress transition is
necessary to make progress out of a location is generated as seen on
line~\ref{alg:2:7} in Algorithm~\ref{alg:2}.

\subsection{Backend code generation}
\label{sec:back-code-gener}

Finally, we describe the last step, step 3, in the compilation procedure
presented in Figure~\ref{fig:1}. After the verification of WHA
requirements and the generation of the SHA, backend code is
generated. In the backend, the SHA is represented as a \emph{Synchronous
  Witness Automata (SWA)}, which captures the behaviour of the SHA as a
synchronous state machine. The SWA is a discrete variant which, when
executed, produces the desired behaviour as a DTTS
(Definition~\ref{def:dtts}).  During the execution, a task known as
\emph{saturation} may also be needed while taking discrete
transitions. This is formalised in
Section~\ref{sec:saturation-function}.  The SWA is formalised using
Definition~\ref{def:swa}.  We start by defining $CV(X[k])$ in
Definition~\ref{def:cvxk} that is essential for the definition of the
transition relation of an SWA.
  
\begin{definition}

Let $x[k]$ denote the $k$-th ($k\in N$) valuation of any variable
 $x \in X$ and let $X[k]$ denote the set of all such $k$-valuations
of variables in $X$.
  
We denote $CV(X[k])$ the set of
  constraints over $X[k]$ : \newline
  $g:=x[k]< n | x[k] \leq n| x[k] > n | x[k] \geq n | g \wedge g$ where,
  $n \in \mathbb{N}$.
  \label{def:cvxk}
\end{definition}

\begin{definition}
  A synchronous witness automata (SWA) \newline
  ${\cal{S}} = \langle S, S_0, \Sigma, X, Updates, \rightarrow \rangle$,
  which corresponds to the \newline
  $SHA = \langle Loc, Edge, \Sigma, Init, Inv, Switness, Jump, Step,$
  $ Nsteps, BoundaryCond \rangle$ where

  \begin{itemize}
  \item $S$ denotes the set of states of $\cal{S}$ and $S=Loc$.
  \item $S_o$ is a subset of initial states.
  \item $\Sigma$ is a set of events.
  \item $X$ is a set of continuous variables 
  \item $Updates$ represent a set of updates. A given update may capture
    the emission of an output event, the update of a state variable in a
    given tick using its witness function or the initialization of an
    integration constant or the time instant.
  \item
    $\rightarrow \subseteq S \times {\cal{B}}(\Sigma) \times CV(X[k])
    \times 2^{Updates} \times S$
    represents the transition relation. Here ${\cal{B}}(\Sigma)$ denotes
    the set of Boolean formulas over $\Sigma$.
  \end{itemize}
  \label{def:swa}
\end{definition}
  
We use the shorthand, $(s,\frac{e \wedge g}{u},s')$ to represent
transitions.  Here, $e \in {\cal{B}}(\Sigma)$, $g \in CV(X[k])$ and
$u \in 2^{Updates}$. 

\ignore{This
transition is enabled when the antecedent (the numerator of the
fraction in the transition relation, also called the \textit{guard}) is
satisfied. When the transition is enabled, the consequent (the
denominator of the fraction in the transition relation, also called the
\textit{updates}) are enabled.}

The generated C-code is an SWA.  The SWA for the water tank generated
from the SHA (Figure~\ref{fig:codeGenSHA}) is shown in
Figure~\ref{fig:codeGenDTTS}. Every location in the SHA has an
equivalent state in the SWA, as shown in
Figure~\ref{fig:codeGenDTTS}. Furthermore, the witness functions in the
locations of the SHA are moved to transitions in the SWA. For example,
the SWA self-transition from state $t_2$ to $t_2$ represents the
evolution of the continuous variable $x$ in discrete steps
$k \in [0, Nsteps\blue{)}$.  In the SWA of Figure~\ref{fig:codeGenDTTS},
we label the transitions with the form:
$\frac{antecedent}{consequent}$. The self-transition from state $t_2$ to
state $t_2$ is labeled as: {
  $\frac{\overline{ON} \wedge \overline{OFF} \wedge (20 \leq x[k] \leq
    100)}{x[k+1]=F_1(\delta,k,C_1), k=k+1} $},
which states that while the value of $x$ at tick $k$ is between $20$ and
$100$ \textit{and} no events (\verb$ON$ or \verb$OFF$) are detected, $x$
evolves to the new value depending upon the witness function
$F_1$. $C_1$ is the constant of integration. The transition guard has a
one-to-one correspondence with the invariant condition on location $t_2$
in Figure~\ref{fig:codeGenSHA}. We add self-transitions to all states of
the generated SWA, which evolve all the continuous variables in the
corresponding location in the SHA. The generated SWA also consists of
all the edges from the SHA, e.g., transition
$(t_1, \frac{{ON} \wedge x[k]=20}{C_1=x[k]-150, x[0]=x[k], k=0},t_2 )$,
which corresponds to the jump and updates on the transition between
locations $t_1$ and $t_2$ in the corresponding SHA in
Figure~\ref{fig:codeGenSHA}.

\begin{figure}[h!]
  \centering
\begin{Verbatim}
double t4_ode_1(double d, double k, double C1){
   return C1*exp(-0.075*d*k);
}
double t4_init_1(double x_u){
   return = x_u;
}
enum states watertankR(enum states cstate, 
                       enum states pstate){
  switch (cstate){
  case (t1):
    if(x == 20 && !ON && !OFF){
      if (pstate != cstate) C1 = t1_init_1(x);
      x_u = t1_ode_1(C1);
      k=k+1;
      cstate = t1;
    }
    else if(!OFF){
      k=1; cstate=t2; x_u = x;
    }
    break;
  case (t2): ...
  case (t3): ...
  case (t4):
    if(x >= 20 && x <= 100 && !ON && !OFF){
      if (pstate != cstate) C1 = t4_init_1(x);
      x_u = t4_ode_1(d, k, C1);
      k=k+1;
      cstate = t4;
    }
    else if(ON){
      k=0; cstate=t2; x_u = x;
    }
    else if(x == 20){
      k=0; cstate=t1; x_u = x;
    }
    break;
  }
  return cstate;
}
\end{Verbatim}
\caption{C-code generated for the SHA in Figure~\ref{fig:codeGen}(c)}
  \label{fig:4}
\end{figure}

The C-code representing the SWA in Figure~\ref{fig:codeGenDTTS} is
shown in Figure~\ref{fig:4}. We have only shown the states $t_1$ and
$t_4$ along with their associated transitions and functions. There is a
literal one-to-one correspondence between the SWA in
Figure~\ref{fig:codeGenDTTS} and the C-code in Figure~\ref{fig:4}. We
use the variables $x$ and $x\_u$ to denote the $x[k]$ and $x[k+1]$
valuations of the continuous variable $x$, respectively. The witness
functions are represented as functions in C (e.g., lines~1-10 in
Figure~\ref{fig:4}).

The witness functions are incomplete, in the sense that the value of the
constant of integration (e.g., $C1$ on line~3) needs to be
computed. Computing this constant of integration is equivalent to
solving the \textit{initial value problem}. It is well known that the
value of the constant of integration, in the witness function, depends
upon the initial value of the witness function. The initial value, at
time $0$, of any witness function is either: (1) the specified initial
value of the continuous variable that the witness function updates, as
shown in Figure~\ref{fig:codeGenDTTS} with the dashed arrow, or (2) the
updated value of the continuous variable on an \texttt{Inter} location
transition.


Consider the two transitions, one from location $t_2$ to $t_4$ and the
other from $t_4$ to $t_2$ in Figure~\ref{fig:codeGenDTTS}. Variable $x$
is updated by the witness function $F_2$ on the self-transition in state
$t_4$. Function $F_2$ takes as input arguments the current tick number,
$k \in \mathbb{N}$, the step size of the tick,
$\delta \in \mathbb{R}^+$, and the constant of integration. The argument
$k$ starts from $0$ and $\delta$ is a constant. We determine the value
of the constant of integration when entering state $t_4$, by making
$C_1$ the subject. Since
$F_2 = C_1 \times e^{\delta \times k \times -0.075}$, we set $k=0$ and
making $C_1$ the subject gives $C_1 = F_2(0)$, i.e., $C_1$ takes the
value of $x[0]$, which as we know from Figure~\ref{fig:codeGenSHA} is
the final value of $x$ in the previous state, since $x' = x$. Hence, we
get $C_1 = x[k]$. In Figure~\ref{fig:4}, the constant of integration
\texttt{C1} is updated whenever any state is entered for the first time
(line~16) this corresponds to the update of the constants of integration
on the transitions in Figure~\ref{fig:codeGenDTTS}.

\subsection{Saturation}
\label{sec:saturation-function}

\begin{figure*}[t]
	\centering
	\begin{minipage}{0.5\linewidth}
	\centering
	
\subfigure[\label{fig:satEx1} Case 1: an increasing function and it does not need saturation]{
\framebox[0.96\textwidth]{
	\begin{tikzpicture}[->,>=stealth',shorten >=1pt,auto,
		node distance=4.4cm,
		semithick,scale=0.9, transform shape]
		\tikzstyle{every state}=[rectangle,rounded corners,
		minimum height = 1.2cm, text width=2.2cm, 
		text centered, fill=blue!20,
		draw=none,text=black, draw,line width=0.3mm]

		\node[state, 
		label={[shift={(0,-1.9)}]$x \leq 120$}, 
		label={[shift={(-1.5,-0.2)}]\large $\mathbf{t_1}$   }]
		(T2) [node distance=4.4cm] 	 {$\dot{x}=0.2x$};

		\node[state, 
		label={[shift={(0,-1.9)}]$x \geq 100$}, 
		label={[shift={(1.5,-0.2)}]\large $\mathbf{t_2}$  }]
		(T4) [right of=T2] 	 {$\dot{x}=0$};
		
		\draw (T2) -- (T4) node [midway, sloped] 
		{\footnotesize $A,x>100$};
		
		
		\draw[<-, dashed](T2.180) -- node[left,text
	       width=1.5cm] {intial $x(0)=50$} ++(-1.2cm,+0cm);
	        
	\end{tikzpicture}
}
}

\subfigure[\label{fig:satEx2} Case 2:  due to equality there is a need for saturation]{
\framebox[0.96\textwidth]{
	\begin{tikzpicture}[->,>=stealth',shorten >=1pt,auto,
	node distance=4.4cm,
	semithick,scale=0.9, transform shape]
	\tikzstyle{every state}=[rectangle,rounded corners,
	minimum height = 1.2cm, text width=2.2cm, 
	text centered, fill=blue!20,
	draw=none,text=black, draw,line width=0.3mm]

	\node[state, fill=gray!20, node distance=3.5cm,
	label={[shift={(0,-1.9)}]$ y<50$}, 
	label={[shift={(-1.5,-0.2)}]\large $ \mathbf{t_3}$ }]
	(B1)[below of= T2]  {$\dot{y}=8.5$};
	
	\node[state, fill=gray!20,
	label={[shift={(0,-1.9)}]$y >= 50$}, 
	label={[shift={(1.5,-0.2)}]\large $\mathbf{t_4}$  }]
	(B4) [right of=B1] 	 {$\dot{y}=0$};

	\draw[<-, dashed](B1.180) -- node[left,text
	width=1.5cm] {intial $y(0)=20$} ++(-1.2cm,+0cm);

	\draw (B1) -- (B4) node [midway] 
	{\footnotesize $\outSignal{\iSignal,y = 50}$};

	\end{tikzpicture}
}
}

\subfigure[\label{fig:satEx3} Case 3:  a decreasing function and it does not need saturation]{
\framebox[0.96\textwidth]{
	\begin{tikzpicture}[->,>=stealth',shorten >=1pt,auto,
	node distance=4.4cm,
	semithick,scale=0.9, transform shape]
	\tikzstyle{every state}=[rectangle,rounded corners,
	minimum height = 1.2cm, text width=2.2cm, 
	text centered, fill=blue!20,
	draw=none,text=black, draw,line width=0.3mm]

	\node[state, fill=red!20, node distance=3.5cm,
	label={[shift={(0,-1.9)}]$ z>-30$}, 
	label={[shift={(-1.5,-0.2)}]\large $ \mathbf{t_5}$ }]
	(C1)[below of= B1]  {$\dot{z}=-5.5$};
	
	\node[state, fill=red!20,
	label={[shift={(0,-1.9)}]$z \leq -10$}, 
	label={[shift={(1.5,-0.2)}]\large $\mathbf{t_6}$  }]
	(C2) [right of=C1] 	 {$\dot{z}=0$};

	\draw[<-, dashed](C1.180) -- node[left,text
	width=1.5cm] {intial $z(0)=0$} ++(-1.2cm,+0cm);

	\draw (C1) -- (C2) node [midway] 
	{\footnotesize $A,z<-10$};

	\end{tikzpicture}
}
}

\end{minipage}
\qquad
\begin{minipage}{0.45\linewidth}
 \subfigure[\label{fig:satPlot} The behaviours of the three example HAs 
 are depicted using solid lines. 
 Our synchronous approximations are depicted 
 using dashed lines. Each tick is one second long.]{
  \framebox[0.96\textwidth]{
	\begin{tikzpicture}[transform shape, xscale=0.93, yscale=1]
	\begin{axis}
	[ xlabel={Time (in seconds)},
	ylabel={Values of $x,y,z$},
	axis y line = left,
	axis x line = bottom,
	xmax=7,
	ymin=-50,   
	ymax=130,
	extra y ticks={120,100,80,50,20,0,-10,-30},
	extra tick style={grid=major}
	]
	
	\addplot[domain=0:4.4, blue, ultra thick,] 
	{50*e^(0.2*x)}; 
	\addplot[blue, ultra thick, smooth,
			domain=4.4:7] (x,120);
			
	\addplot[color=red,fill=red,only marks,mark=*] coordinates{(0,50)(1,61.07)(2,74.59)(3,91.11)(4,111.27)(5,120)(6,120)(7,120)};

	\addplot[gray, ultra thick, smooth,
			domain=0:3.52] (x,8.5*x+20);
	\addplot[gray, ultra thick, smooth,
		domain=3.5:7] (x,50);

	\addplot[color=red,fill=red,only marks,mark=*] 
	coordinates{(0,20)(1,28.5)(2,37)(3,45.5)(4,50)(5,50)(6,50)(7,50)};
			
	\addplot[red, ultra thick, smooth,
			domain=0:5.45] (x,-5.5*x);
	\addplot[red, ultra thick, smooth,
			domain=5.45:7] (x,-30);

	\addplot[color=red,fill=red,only marks,mark=*] 
	coordinates{(0,0)(1,-5.5)(2,-11)(3,-16.5)(4,-22)(5,-27.5)(6,-30)(7,-30)};

	
	%
		 
	\node[pin=170:{ (4,111.2)}] at (axis cs:4,111) {};
	
	\node[pin=180:{ (3,45.5)}] at (axis cs:3,45.5) {};
	\end{axis}
	\end{tikzpicture}
	}
}
\end{minipage}
	\caption{The need for saturation depends on the
	location invariant, the guard in HA, and the step size. 
	Out of  the three cases, only Case~2 requires saturation,
	see  Figure~\ref{fig:satPlot}. }
	\label{fig:saturation}
\end{figure*}
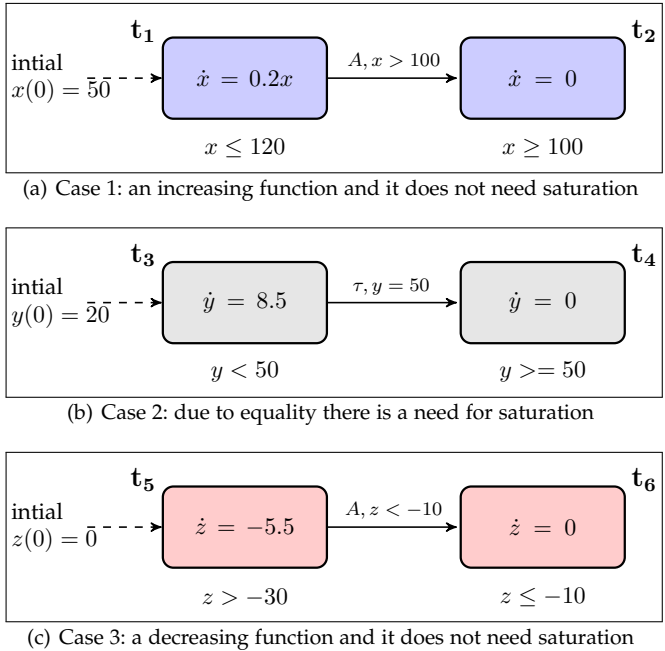
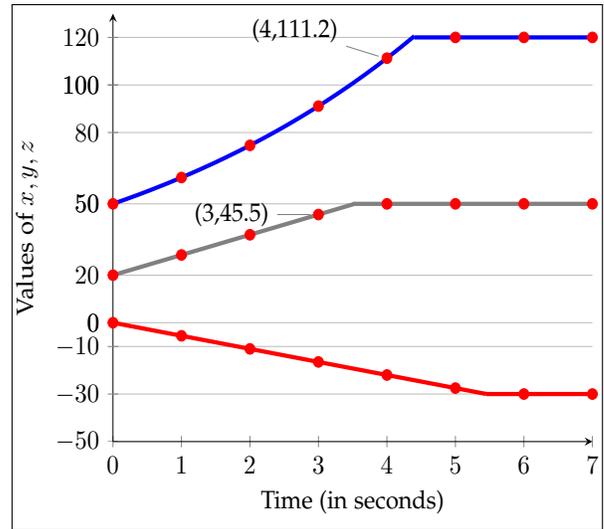

Due to the discrete/synchronous valuation of variables, when control
remains in one location the variable may
reach a value which will never satisfy the guard condition of the egress
transition. This phenomena depends on the location invariant and
the guard and may not happen frequently. Figure~\ref{fig:saturation}
illustrates three separate cases.

\noindent \textbf{Case~1}: Consider the step size $\delta$ as 1,
 the invariant at
location $t_1$ in Figure~\ref{fig:satEx1}  is $x\leq 120$ and the jump condition is
$A \wedge x >100$. 
According to the non-deterministic 
semantics of hybrid automata,
the value of $x$, as it leaves location $t_1$, 
can be in the interval $(100, 120]$ when the signal~$A$ is present.
The value of $x$ is plotted against 
time in Figure~\ref{fig:satPlot} (top).
In the synchronous approximation (shown as points),
we have two scenarios.
(1) If the transition is taken at the $5^{th}$ tick 
(time equals $4$ seconds), 
the value of $x$ as it leaves location~$t_1$ 
is $111.2$, respectively.
(2) Otherwise, at the $6^{th}$ tick (time equals $5$ ) 
the  value of $x$ is saturated to $120$
and is forced to exit location~$t_1$.
Note that the trace due to the synchronous approximation
is different but, still a valid trace of the \ac{HA}. 

This example satisfies the
	 restrictions proposed by Alur~et al.~\cite{alur2003generating}
	 and will be accepted for code generation.
	 Here, the duration between the occurrence of
	 the jump condition evaluating to true and the occurrence of the 
	 invariant evaluating to false is longer than
    the tick length (sampling period).
    This restriction ensures that there is always 
    at least one valid state where a discrete transition can be taken.
    However, the above restriction may not always hold as illustrated 
    in the following case.

\noindent \textbf{Case~2}: Given the step size as 1, 
Figure~\ref{fig:satEx2} illustrates a 
case where  we must saturate. 
Figure~\ref{fig:satPlot} shows the value of $y$ increasing steadily
to $50$ and then remaining at $50$. 
In contrast to Case~1,
 for location $t_3$ the duration between 
 the guard ($y=50$) evaluating to true and the occurrence of the
  invariant evaluating to false ($y<50$) is zero.
  This violates the restrictions proposed in~\cite{alur2003generating}
  and this HA would not be accepted in their framework. 
  In Simulink\textsuperscript{\textregistered}, this HA may not be simulated correctly 
  because of the equality check. In fact, for two out of seven 
  benchmarking examples, we observed incorrect behaviour in Simulink\textsuperscript{\textregistered}
  (see Table~\ref{tab:vsAlur} in Section~\ref{sec: results}).
  
  In the synchronous approximation, we accept this HA 
  because of the proposed \emph{saturation} technique.
  In our approach, during the $4^{th}$ tick the value of $y$ 
  is below $50$ and the guard condition is not satisfied. 
  During the $5^{th}$ tick the value of $y$ is above $50$ 
  and without saturation
  this leads to an unreachable state because neither the
   invariant ($y < 50$) of  $t_3$ nor the guard
   ($y=50$) are valid.
   To address this problem, at the start of the $5^{th}$ tick,
    we saturate the
   value of $y$ to exactly $50$. Hence, a value greater
   than $50$ is not observed between ticks $4$ and~$5$.
   By using saturation, unlike~\cite{alur2003generating}, 
   we are agnostic to the step size (sampling period).

\noindent \textbf{Case~3}: Similar to Case~1,
Figure~\ref{fig:satEx3} illustrates a 
decreasing function.
Once again, the duration between the occurrence of 
the guard evaluating to true and the occurrence of the 
 invariant evaluating to false is longer than
 the step size (sampling period). This restriction ensures that there always exists 
 at least one valid state where a discrete transition can be taken.
This HA is accepted for code generation by~\cite{alur2003generating}
and by our tool~\ourTool.

 Definition~\ref{def:saturation} formalises this \emph{saturation} technique.
 

\begin{definition}
  In any discrete time instant $k$, when execution makes a discrete
  switch from state $l$ to $l'$ in the SWA, the valuation of all
  continuous variables $X[k]$ either satisfy the guard condition
  $g \in CV(X[k])$ or are set to a suitable value such that $g$ is
  satisfied. This is termed as \emph{saturation}.
  \label{def:saturation}
\end{definition}

According to the above definition, during the execution of the backend code,
we must decide dynamically when to saturate and also must decide on the
correct value. This decision is based on the following Lemma that
ensures that the value of $x[k]$ that satisfies the guard always exists
in the current discrete step.

\begin{lemma}
  It is always possible to uniquely determine the saturation value for
  any continuous variable at time instant $k$ when the state (location)
  switch from $l$ to $l'$ is to be taken in a SWA.
\label{le:saturation}
\end{lemma}


\begin{proof}
The proof of this lemma follows from the following observations.

\begin{itemize}
\item Observation 1: All witness functions $x(t)$ for any $x \in X$ are
  monotonic in every location (WHA requirement).
\item Observation 2: All witness functions $x(t)$ for any $x \in X$ are
  continuous as they are differentiable in any interval.
\item Observation 3: Given the above two observations, the saturation
  value for any variable $x$ always exists in the time interval
  $[(k-1) \times \delta, k \times \delta]$ when the location switch
  happens at instant $k \times \delta$.
\end{itemize}
\end{proof}

\subsection{Composition of SHAs}
\label{sec:composition-sha}

\begin{figure*}[h!]
  \centering
  \input{./figures/parallelComposition}
  \caption{Composition of \ac{HA}s in our framework}
  \label{fig:parallelComp}
\end{figure*}

This section describes the modular compilation of SHAs using the
composition rules defined in Definition~\ref{sec:sha-semant-comp}.
Let us consider a part of the water tank and gas burner HAs as shown in
Figure~\ref{fig:parcomHA}. The equivalent partial SHAs for these two HAs
are shown in Figure~\ref{fig:parcomSHA}. As described previously, in
Section~\ref{sec:generation-sha}, the ODEs (flow predicates) in each
location have been replaced with their individual witness functions. We
can furthermore compile these SHAs into individual SWAs as shown in
Figure~\ref{fig:parcomDTTS}. The composition rules defined in
Definition~\ref{sec:sha-semant-comp} are applied on these resultant SWAs
to generate a single SWA shown in Figure~\ref{fig:parcomCode}.

\begin{algorithm}[t!]
  \begin{algorithmic}[1]
    \REQUIRE SWA $Q_1$, $Q_2$
    \STATE $Q_c \leftarrow Q_1 \times Q_2$ \label{alg:3:1} \\
    \COMMENT{Apply \SHArule{Intra-Intra}}
    \FORALL{$q \in Q_c$}  \label{alg:3:2}
    \STATE build\_self\_transition(q)
    \ENDFOR \\ \label{alg:3:3}
    \COMMENT{Apply \SHArule{Intra-Inter}}
    \FORALL{$q_1 \in Q_c $} \label{alg:3:4}
    \STATE $(l_1, l_2) \leftarrow q_1$ \label{alg:3:6}
    \FORALL{$q_2 \in Q_c \wedge q_1 \neq q_2$}
    \STATE $(l_1', l_2') \leftarrow q_2$ \label{alg:3:9}
    \IF{$l_1 = l_1' \wedge l_2 \neq l_2'$} \label{alg:3:7}
    \STATE build\_transition($q_1, q_2$) \label{alg:3:8}
    \ENDIF
    \ENDFOR
    \ENDFOR \label{alg:3:5}\\
    \COMMENT{\SHArule{Inter-Intra} is similar to the
      \SHArule{Intra-Inter}}
  \end{algorithmic}
  \caption{The pseudo-code used to compose two SWAs}
  \label{alg:3}
\end{algorithm}

The pseudo-code used to compose two SWAs is shown in
Algorithm~\ref{alg:3}. The algorithm takes as input two SWAs: $Q_1$ and
$Q_2$, respectively. In case of $N$ SWAs, the composition procedure is
applied recursively. As the very first step, the algorithm builds the
cross product of all states in the constituent SWAs $Q_1$ and $Q_2$
(line~\ref{alg:3:1}). Next, the intra-location and inter-location rules
from Definition~\ref{sec:sha-semant-comp} are applied sequentially.

The very first rule that we apply is the intra-location transition rule
(\SHArule{Intra-Intra}). The application of this rule simply requires one to
iterate through each state in the product set $Q_c$, building
self-transitions on states. The example of the application of
\SHArule{Intra-Intra} is shown in Figure~\ref{fig:parcomCode}. Given state
$(t_2, b_4)$ in the resultant SWA; a self-transition from $(t_2, b_4)$
to $(t_2, b_4)$ is introduced. This transition indicates the
simultaneous evolution of the ODEs from individual states $t_2$ and
$b_4$ from the constituent SWAs $Q_1$ and $Q_2$, respectively. As seen
from Figure~\ref{fig:parcomCode}, the continuous variables $x$ and $y$
evolve together according to their individual witness functions from the
self-transitions on states $t_2$ and $b_4$ in Figure~\ref{fig:parcomDTTS},
provided that the \textit{conjunction} of the individual guards holds.

Next, we apply the three inter-location rules
sequentially. Algorithm~\ref{alg:3} only shows the application of
\SHArule{Intra-Inter}, since other rules can be trivially derived from
this rule. The algorithm traverses through each state of the the product
set $Q_c$. Upon visiting a state $q_1$, the state label is first
decomposed into its constituent parts (line~\ref{alg:3:6}). For the
running example, given $q_1 = (t_2, b_4)$, line~\ref{alg:3:6}, gives
$l_1 = t_2$ and $l_2 = b_4$, respectively. Next, we iterate through all
the states in $Q_c$ other than state $q_1$, again decomposing the state
label into its constituent location names, $l_1'$ and $l_2'$,
respectively at line~\ref{alg:3:9}. \SHArule{Intra-Inter} states that
the second SWA $Q_2$ makes an inter-location transition. SWA $Q_1$ on
the other hand is forced to make a transition, such that the destination
state after the transition has the same location label as the source
transition state. Hence, the algorithm builds transitions from state
$q_1$ to any state $q_2$ such that the constituent label $l_1$ of state
$q_1$ and $l_1'$ of state $q_2$ are the same, but $l_2$ and $l_2'$ are
different. The result of application of such a rule is shown in
Figure~\ref{fig:parcomDTTS}, as transition
($t_2b_4, \frac{\overline{ON} \wedge \overline{OFF} \wedge (0 \leq x[k]
  \leq 100) \wedge (y[k]=\frac{1}{10}) } {y[k+1]=y[k], x[0] = x[k], OFF,
  k=0},t_2b_1$).
The second SWA $Q_2$ takes a discrete transition, forcing the first one
to also take a discrete transition. This transition is only to a state
where the location of the first SWA ($Q_1$) does not change. The guards
on this transition are a conjunction of the set of individual guards. We
add a special update $x[0] = x[k]$, which carries the value of the
continuous variable $x$. Note that the update $x[0] = x[k]$ is the
consequence of $i_1' = v_1$ in \SHArule{Intra-Inter}.


\section{Results}
\label{sec: results}

We compare the efficacy of the proposed approach 
(our tool \ourTool) 
with Simulink. In particular we compare the performance 
of these two tools relative to execution time and 
code size.
For the purposes of this comparison, we use the seven benchmarks presented in
Table~\ref{tab:benchmarks}. 
These benchmarks span across different application
 domains such as medical, physics, and 
 industrial automation, illustrating the diversity of the proposed 
  approach. In our setting we mean
  IOHA~\cite{lynch03} when we state HA. IOHA enable the modelling of the plant and
  the controller separately.

 As depicted in column one of Table~\ref{tab:benchmarks}, four out of the 
 seven benchmarks are described using a single \ac{HA} and 
 the remaining three examples are described using more than one \ac{HA}.
 The table also presents the 
 number of locations (\#L) in each hybrid automata.
 For example, (2,3) denotes that the Train Gate Control
 (TG) benchmark is described using a HA with two locations 
 and a second HA with three locations. More details
 about the benchmarks and their 
 implementation in \ourTool and Simulink
 are available online~\cite{githubBenchmarks}.
 

\subsection{Experimental set-up} 
The following steps are considered in order to
 achieve a fair comparison between \ourTool and Simulink.
 
 \begin{description}
 	\item[\textbf{Solver}] To reflect the synchronous execution model,
	 	we used a discrete solver with a fixed step in Simulink.
	 	The Discrete-Time Integrator block is configured to
        use the Forward Euler method. Other methods such as 
        Backward Euler and Trapezoidal resulted in an ``algebraic loop''
        error and we did not pursue a solution to this.
	 	
	\item[\textbf{Step size}] 	For all benchmarks the step size in Simulink is fixed to $0.01$ seconds.
	Also the same step size is used in \ourTool, $\delta = 0.01$ seconds.
	
	\item[\textbf{Time}] All benchmarks were simulated in Simulink for 
	$100,000$ seconds of simulation time. Based on a step size of $0.01$ seconds,
	in \ourTool this translates to $10$ million ticks.
	 	
 \end{description}

The experiments are executed on an Intel~i7 processor
with 16~GB RAM running the Windows~7 operating system. 

\subsection{Evaluation}

For all the benchmarks, the executable
for the Simulink models  are generated using the 
in-built C~code generator. It automatically
generates equivalent C~code and compiles it to produce
an executable.
Similarly, \ourTool generates equivalent C~code
and generates an executable using GCC.
 The execution time and the code size of the
 generated executables are reported below.
 
\textbf{Execution time:} 
Figure~\ref{fig:compareET} 
shows that for all benchmarks, 
the execution times of \ourTool are significantly 
shorter than Simulink.
On average, \ourTool is $3.9$ times faster than Simulink. 
The most significant difference between the tools
is observed for the most complex example 
(most locations), the
 water tank heating system~(WH). 
 For this example, the execution
time of \ourTool is $7.3$ times faster 
than Simulink. 


\textbf{Code size:} Figure~\ref{fig:compareET} 
shows that for all benchmarks, 
the code size of \ourTool is significantly 
smaller than Simulink.
On average, the generated code is $40\%$ 
smaller than the Simulink code.
The most significant difference is $47\%$ which is observed
for the Thermostat~(TS) example.
In general, Simulink is a more feature rich tool
and there may be some overheads during code generation
where numerical solvers are linked during compilation.


In summary, on average Piha is faster 
(in execution time) than Simulink by a factor of $3.9$ times 
and the generated code is $40\%$ smaller than the Simulink code.
In the future, we will extend our benchmark suite with more complex 
examples and quantitatively compare with Ptolemy and Z\'elus.

Finally, in Table~\ref{tab:vsAlur} we qualitatively compare the 
acceptability of the benchmarks between our tool and the tool based on~\cite{alur2003generating}. 
As discussed earlier in the introduction and Section~\ref{sec:saturation-function},
the tool in~\cite{alur2003generating} requires that ``a given mode and the corresponding guard of any switch must
overlap for a duration that is greater than the sampling
period''~\cite{alur2003generating}. Due to this restriction,
it is not possible to accept any HA that has a guard condition that
checks for equality. For the running example (Figure~\ref{fig:waterTankHAtank}), 
the guard condition $x=100$ checks if the water has reached the boiling 
point. Similarly, in the train gate control benchmark,
there is a guard condition that checks if the gate height is exactly 
equal to $10$~meters. In summary, due to our saturation function,
we can accept a larger set of benchmarks for code generation than the
tool in~\cite{alur2003generating}.

\begin{table*}
	\centering
	\caption{Benchmark descriptions
	\label{tab:benchmarks}}
\begin{tabular}{| c | c | c| c| l  |} \hline
\multicolumn{2}{|c|}{\textbf{Benchmarks}} 
	& \textbf{Domain} 
	& \textbf{\# L } 
	& \textbf{Description} \\ \hline
\multirow{4}{*}{\centering \rotatebox{90}{Single} \rotatebox{90}{HA}}
	& Thermostat (TS) 
		& Physics~\cite{Pedro2005}
		& 2
		& Heats a room to keep it warm.\\ \cline{2-5}
	& Switch Tank (ST)  
		& Physics~\cite{lygeros2008hybrid}
		& 2
		& A hose adding water by switching between two leaky tanks. \\ \cline{2-5}
	& Heart Cell (HC)  
		& Biology~\cite{chen201487}
		& 4
		& Captures the electrical behaviour of a cardiac cell.\\ \cline{2-5}
	& Train Brake control (TB) 
		& Industrial automation~\cite{platzer2012logical} 
		& 2
		& Maintains a train's speed between the upper and lower limit.\\ \hline
\multirow{4}{*}{\centering \rotatebox{90}{Two} \rotatebox{90}{HAs}}
	& Water Heating system (WH)  
		& Physics~\cite{raskin05}
		& 4 , 4
		& Models the heating of water (see Figure~\ref{fig:waterTank}) \\ \cline{2-5}
	& Train Gate control(TG)  
		& Industrial automation~\cite{Costello2013}
		& 2, 3
		& Models the behaviour of a gate at a rail road crossing.\\ \cline{2-5}
	& Nuclear Plant control (NP)  
		& Industrial automation~\cite{alur2015principles}
		& 3, 3
		& Switches between two fuel rods to avoid a meltdown.\\ \hline

 \end{tabular}
 \end{table*}
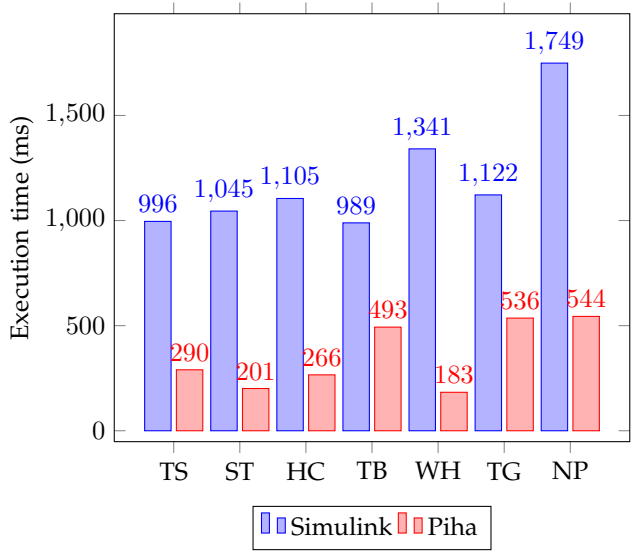
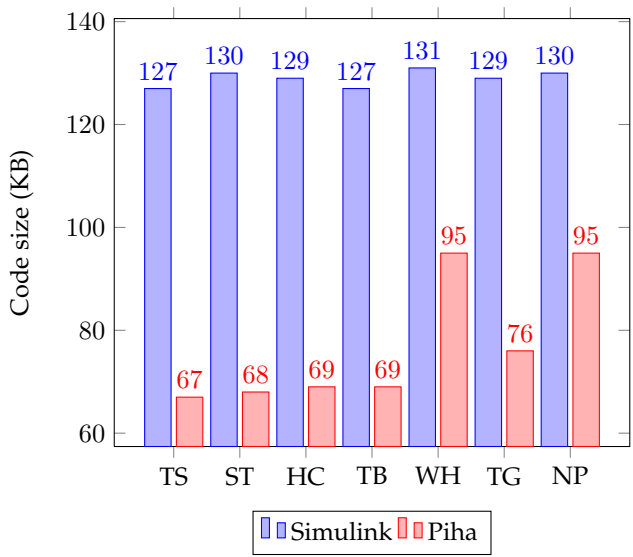
\begin{figure}[htbp]
	\centering
	\begin{minipage}{1.95\linewidth}
		\subfigure[Execution time (ms) \label{fig:compareET}]{
		
		\begin{tikzpicture}
		\begin{axis}[
		ybar,
		enlargelimits=0.15,
		legend style={at={(0.5,-0.15)},
			anchor=north,legend columns=-1},
		ylabel={Execution time (ms)},
		symbolic x coords={TS, ST, HC, TB, WH, TG, NP},
		xtick=data,
		nodes near coords,
		nodes near coords align={vertical},
		]
		\addplot coordinates {
			(TS,996)
			(ST,1045)
			(HC,1105)
			(TB,989)
			(WH,1341)
			(TG,1122)
			(NP,1749) };
		\addplot coordinates {
			(TS,290)
			(ST,201)
			(HC,266)
			(TB,493)
			(WH,183)
			(TG,536)
			(NP,544) };
		\legend{Simulink, Piha}
		
		\end{axis}
		\end{tikzpicture}

		}
		
		\qquad
		
		\subfigure[Code size (KB) \label{fig:compareCS}]{
			
			\begin{tikzpicture}
			\begin{axis}[
			ybar,
			enlargelimits=0.15,
			legend style={at={(0.5,-0.15)},
				anchor=north,legend columns=-1},
			ylabel={Code size (KB)},
			symbolic x coords={TS, ST, HC, TB, WH, TG, NP},
			xtick=data,
			nodes near coords,
			nodes near coords align={vertical},
			]
			\addplot coordinates {
				(TS,127)
				(ST,130)
				(HC,129)
				(TB,127)
				(WH,131)
				(TG,129)
				(NP,130) };
			\addplot coordinates {
				(TS,67)
				(ST,68)
				(HC,69)
				(TB,69)
				(WH,95)
				(TG,76)
				(NP,95) };
			\legend{Simulink,Piha}
			
			\end{axis}
			\end{tikzpicture}

		}
	\end{minipage}

	\caption{
		Comparing the execution time (in ms) and code size (in KB) of Simulnk 
		and \ourTool~for the benchmarks in Table~\ref{tab:benchmarks}.
	}
	\label{fig:results}
\end{figure}
\begin{table}
  \centering
  \caption{What benchmarks can be accepted based on the  restrictions of the tool in~\cite{alur2003generating} and  our tool \ourTool. \label{tab:vsAlur}}
  \begin{tabular}{| l | c | c| } \hline
    \textbf{Benchmarks} & \textbf{Tool in~\cite{alur2003generating}} & \textbf{\ourTool} \\ \hline

    Thermostat (TS) 			& Yes	& Yes \\ \hline
    Switch Tank (ST) 			& Yes	& Yes \\ \hline
    Heart Cell (HC) 			& Yes	& Yes \\ \hline
    Train Brake control (TB) 	& Yes	& Yes \\ \hline
    Water Heating system (WH)   & No	& Yes \\ \hline
    Train Gate control(TG)  	& No	& Yes \\ \hline
    Nuclear Plant control (NP)  & Yes	& Yes \\ \hline

  \end{tabular}
\end{table}


\section{Conclusions}
\label{sec:conclusions}

Hybrid automata (HA)~\cite{alur93} is a very well known framework for the modelling
and verification of \acf{CPS}~\cite{alur2015principles}. The primary focus of the current work is 
the emulation of controllers using plant models that provide real-time
 closed-loop response, while validating the controllers in a CPS. 
We term such plant models as plant-on-a-chip
 (PoC).

 HA was initially developed for the
simulation / verification of CPS where both the plant and the controller
are considered as a single HA. We seek to validate controllers 
 using emulation of the plant. Variants such as input / output
 HA~\cite{lynch03} (IOHA) are
proposed to specify both the plant and the adjoining controller and
this work has recently been used for the validation of controllers
such as pacemakers~\cite{chen14}. The majority of controller validation 
 approaches use tools such as Simulink for the modelling and code
 generation of the plant to validate the controller. There are well
 known semantic limitations of using such tools for
validation. Tools such as Ptolemy~\cite{ptolemaeus2014system} and
Z\'elus~\cite{bourke13zelus}, on the other hand, are founded on formal semantics. However,
they have limitations for emulating plants due to the dynamic
interaction with numerical solvers. 
There have also been some prior work on automatic code generation from
HA models~\cite{alur2003generating, kim2003modular}.
However, these approaches handle very restrictive set of examples.

This paper formulates the problem of emulation of 
\ac{CPS} using automated algorithmic techniques
for code generation from IOHA models. We firstly defined a set of
well-formedness criteria that are specifically developed for
facilitating code generation. We also propose a discrete time
semantics of well-formed HA (WHA) based on the synchronous
approach~\cite{benveniste03}. Based on this semantics and an approach 
 for synchronous composition of HA models, also proposed here, we are
 able to perform modular compilation from a network of HA models to C
 code for PoC design. 

Experimental validation of the proposed approach compared with
Simulink reveals that the generated code is efficient both in terms of
execution time and code size while avoiding the semantic ambiguities
associated with Simulink. The proposed approach, thus, paves the way
for the emulation of a wide range of CPS applications in automotive,
medical devices, and robotics. 

In the near future, we will compare Piha with Z\'elus and Ptolemy quantitatively. 
We will also pursue the decidability question of WHA.

\ignore{While this paper is a starting point for PoC design, several open
questions need to be addressed in the near future. Firstly, we will
examine automatic techniques to transform HA descriptions to
WHA models, when feasible. We will also explore automated techniques
for reachability analysis of WHA models and also consider the issue of
decidability. Finally, we will explore code generation from a broader
class of HA models which may not satisfy the monotonicity requirements
and may involve dependent ODEs. We also plan to explore \ac{HA} to 
hardware synthesis to exploit the use of true parallelism for emulating
plant.}

\section*{Acknowledgements}

This research was funded by the University of Auckland FRDF
Postdoctoral grant No. 3707500. The authors would like to thank Th\'eo
Steger, an intern from the ENSTA Paris Tech University, for
his help in developing benchmarks and testing our tool.
The authors also acknowledge discussions with colleagues in
Control Engineering, in particular Nitish Patel and Akshya
Swain for discussions on sampled systems and synchronous
compositions between plant and controllers.

\ifCLASSOPTIONcaptionsoff
  \newpage
\fi



%


\begin{thebibliography}{10}

\bibitem{lee08cps}
E.~A. Lee, ``{Cyber Physical Systems: Design Challenges},'' in {\em
  {Proceedings of the 2008 11th IEEE Symposium on Object Oriented Real-Time
  Distributed Computing}}, ISORC '08, (Washington, DC, USA), pp.~363--369, IEEE
  Computer Society, 2008.

\bibitem{alur2015principles}
R.~Alur, {\em {Principles of Cyber-Physical Systems}}.
\newblock MIT Press, 2015.

\bibitem{raskin05}
J.-F. Raskin, {\em Handbook of Networked and Embedded Control Systems}, ch.~An
  introduction to hybrid automata, pp.~491--517.
\newblock Springer, 2005.

\bibitem{alur93}
R.~Alur, C.~Courcoubetis, T.~A. Henzinger, and P.-H. Ho, ``{Hybrid Automata: An
  Algorithmic Approach to the Specification and Verification of Hybrid
  Systems},'' in {\em Hybrid Systems}, (London, UK), pp.~209--229,
  Springer-Verlag, 1993.

\bibitem{wilhelm08wcrt}
R.~Wilhelm, J.~Engblom, A.~Ermedahl, N.~Holsti, S.~Thesing, D.~Whalley,
  G.~Bernat, C.~Ferdinand, R.~Heckmann, T.~Mitra, F.~Mueller, I.~Puaut,
  P.~Puschner, J.~Staschulat, and P.~Stenstr\"{o}m, ``The worst-case
  execution-time problem---overview of methods and survey of tools,'' {\em
  Trans. on Embedded Computing Systems, ACM}, vol.~7, no.~3, pp.~1--53, 2008.

\bibitem{iec61508}
{International Electrotechnical Commission}, ``{IEC 61508 Functional safety of
  electrical / electronic / programmable electronic safety-related systems}.''
  http://www.iec.ch/functionalsafety/.
\newblock last accessed 01/09/2015.

\bibitem{iso26262}
{International Organization for Standardization}, ``{ISO 26262-1: Road vehicles
  -- Functional safety}.''
  http://www.iso.org/iso/catalogue\char`_detail?csnumber=43464.
\newblock last accessed 01/09/2015.

\bibitem{youn2015software}
W.~K. Youn, S.~B. Hong, K.~R. Oh, and O.~S. Ahn, ``{Software certification of
  safety-critical avionic systems: DO-178C and its impacts},'' {\em Aerospace
  and Electronic Systems Magazine, IEEE}, vol.~30, no.~4, pp.~4--13, 2015.

\bibitem{fda}
``{US Food and Drug Administration (FDA)}.'' http://www.fda.gov.
\newblock last accessed 29/05/2015.

\bibitem{alemzadeh13}
H.~Alemzadeh, R.~K. Iyer, Z.~Kalbarczyk, and J.~Raman, ``Analysis of
  safety-critical computer failures in medical devices,'' {\em Security \&
  Privacy, IEEE}, vol.~11, no.~4, pp.~14--26, 2013.

\bibitem{alur08symbolic}
R.~Alur, A.~Kanade, S.~Ramesh, and K.~Shashidhar, ``{Symbolic analysis for
  improving simulation coverage of Simulink/Stateflow models},'' in {\em
  Proceedings of the 8th ACM international conference on Embedded software},
  pp.~89--98, ACM, 2008.

\bibitem{bourke13zelus}
T.~Bourke and M.~Pouzet, ``{Z{\'e}lus: a synchronous language with ODEs},'' in
  {\em Proceedings of the 16th international conference on Hybrid systems:
  computation and control}, pp.~113--118, ACM, 2013.

\bibitem{ptolemaeus2014system}
C.~Ptolemaeus, {\em System Design, Modeling, and Simulation: Using Ptolemy II}.
\newblock Ptolemy. org, 2014.

\bibitem{bourke15code}
T.~Bourke, J.-L. Colaco, B.~Pagano, C.~Pasteur, and M.~Pouzet, ``A
  synchronous-based code generator for explicit hybrid systems languages,'' in
  {\em Compiler Construction} (B.~Franke, ed.), vol.~9031 of {\em Lecture Notes
  in Computer Science}, pp.~69--88, Springer Berlin Heidelberg, 2015.

\bibitem{patel2015survey}
K.~N. Patel and R.~H. Javeri, ``A survey on emulation testbeds for mobile
  ad-hoc networks,'' {\em Procedia Computer Science}, vol.~45, pp.~581--591,
  2015.

\bibitem{benveniste03}
A.~Benveniste, P.~Caspi, S.~Edwards, N.~Halbwachs, P.~Le~Guernic, and
  R.~de~Simone, ``The synchronous languages 12 years later,'' {\em Proceedings
  of the IEEE}, vol.~91, pp.~64--83, Jan. 2003.

\bibitem{baier08}
C.~Baier and J.-P. Katoen, {\em {Principles of Model Checking}}.
\newblock The MIT Press, 2008.

\bibitem{SpaceEx}
G.~Frehse, C.~Le~Guernic, A.~Donz{\'e}, S.~Cotton, R.~Ray, O.~Lebeltel,
  R.~Ripado, A.~Girard, T.~Dang, and O.~Maler, ``{SpaceEx: Scalable
  Verification of Hybrid Systems},'' in {\em Proceedings of the 23rd
  International Conference on Computer Aided Verification}, CAV'11, (Berlin,
  Heidelberg), pp.~379--395, Springer-Verlag, 2011.

\bibitem{tripakis05translating}
S.~Tripakis, C.~Sofronis, P.~Caspi, and A.~Curic, ``{Translating discrete-time
  Simulink to Lustre},'' {\em ACM Transactions on Embedded Computing Systems
  (TECS)}, vol.~4, no.~4, pp.~779--818, 2005.

\bibitem{andalam14}
S.~Andalam, P.~S. Roop, A.~Girault, and C.~Traulsen, ``A predictable framework
  for safety-critical embedded systems,'' {\em IEEE Transactions on Computers},
  vol.~63, no.~7, pp.~1600--1612, 2014.

\bibitem{scade}
``{SCADE Tools}.'' http://www.esterel-technologies.com/.
\newblock last accessed - 19.6.15.

\bibitem{fornari2010understanding}
X.~Fornari, ``{Understanding how SCADE suite KCG generates safe C code},'' {\em
  White paper, Esterel Technologies}, 2010.

\bibitem{alur2003generating}
R.~Alur, F.~Ivancic, J.~Kim, I.~Lee, and O.~Sokolsky, ``Generating embedded
  software from hierarchical hybrid models,'' {\em ACM SIGPLAN Notices},
  vol.~38, no.~7, pp.~171--182, 2003.

\bibitem{kim2003modular}
J.~Kim and I.~Lee, ``Modular code generation from hybrid automata based on data
  dependency,'' in {\em Real-Time and Embedded Technology and Applications
  Symposium, 2003. Proceedings. The 9th IEEE}, pp.~160--168, IEEE, 2003.

\bibitem{Bresolin2011}
D.~Bresolin, L.~Di~Guglielmo, L.~Geretti, and T.~Villa,
  ``{Correct-by-construction code generation from hybrid automata
  specification},'' in {\em International Wireless Communications and Mobile
  Computing Conference (IWCMC)}, pp.~1660--1665, IEEE, July 2011.

\bibitem{alur94}
R.~Alur and D.~L. Dill, ``A theory of timed automata,'' {\em Theoretical
  computer science}, vol.~126, no.~2, pp.~183--235, 1994.

\bibitem{yoong2015model}
L.~H. Yoong, P.~S. Roop, Z.~E. Bhatti, and M.~M. Kuo, {\em {Model-Driven Design
  Using IEC 61499}}.
\newblock Springer, 2015.

\bibitem{chen14}
T.~Chen, M.~Diciolla, M.~Kwiatkowska, and A.~Mereacre, ``Quantitative
  verification of implantable cardiac pacemakers over hybrid heart models,''
  {\em Information and Computation}, vol.~236, pp.~87--101, 2014.

\bibitem{lynch03}
N.~Lynch, R.~Segala, and F.~Vaandrager, ``{Hybrid I/O automata},'' {\em
  Information and computation}, vol.~185, no.~1, pp.~105--157, 2003.

\bibitem{brooks15}
C.~Brooks, E.~A. Lee, D.~Lorenzetti, T.~S. Nouidui, and M.~Wetter, ``{CyPhySim:
  A Cyber-physical Systems Simulator},'' in {\em Proceedings of the 18th
  International Conference on Hybrid Systems: Computation and Control}, HSCC
  '15, (New York, NY, USA), pp.~301--302, ACM, 2015.

\bibitem{harel:ReactiveSystems}
D.~Harel and A.~Pnueli, ``{On the Development of Reactive Systems},'' in {\em
  {Logics and Models of Concurrent Systems}} (K.~Apt, ed.), NATO ASI Series,
  Vol. F-13, (La Colle-sur-Loup, France), pp.~477--498, Springer-Verlag, 1985.

\bibitem{ogata10}
K.~Ogata, {\em Modern control engineering}.
\newblock Boston : Prentice-Hall, 2010.

\bibitem{carlsson2012methods}
H.~Carlsson, B.~Svensson, F.~Danielsson, and B.~Lennartson, ``{Methods for
  reliable simulation-based PLC code verification},'' {\em Industrial
  Informatics, IEEE Transactions on}, vol.~8, no.~2, pp.~267--278, 2012.

\bibitem{ju:PerformanceEsterel}
L.~Ju, B.~K. Huynh, A.~Roychoudhury, and S.~Chakraborty, ``{Performance
  Debugging of Esterel Specifications},'' in {\em {IEEE/ACM/IFIP International
  Conference on Hardware/Software Codesign and System Synthesis (CODES+ISSS)}},
  (Atlanta), pp.~173--178, ACM, October 2008.

\bibitem{roop:TightWCRT}
P.~S. Roop, S.~Andalam, R.~von Hanxleden, S.~Yuan, and C.~Traulsen, ``{Tight
  {WCRT} Analysis of Synchronous {C} Programs},'' in {\em {Proceedings of the
  international conference on Compilers, architecture, and synthesis for
  embedded systems (CASES)}}, (Grenoble), pp.~205--214, ACM, October 2009.

\bibitem{Wang13ILP}
J.~J. Wang, P.~S. Roop, and S.~Andalam, ``{ILPc: A Novel Approach for Scalable
  Timing Analysis of Synchronous Programs},'' in {\em {International Conference
  on Compilers, Architecture, and Synthesis for Embedded Systems (CASES)}},
  Oct. 2013.

\bibitem{berry2000foundations}
G.~Berry, ``{The foundations of Esterel.},'' in {\em Proof, language, and
  interaction}, pp.~425--454, 2000.

\bibitem{halbwachs1991synchronous}
N.~Halbwachs, P.~Caspi, P.~Raymond, and D.~Pilaud, ``{The synchronous data flow
  programming language LUSTRE},'' {\em Proceedings of the IEEE}, vol.~79,
  no.~9, pp.~1305--1320, 1991.

\bibitem{leguernic1991programming}
P.~LeGuernic, T.~Gautier, M.~Le~Borgne, and C.~Le~Maire, ``{Programming
  real-time applications with SIGNAL},'' {\em Proceedings of the IEEE},
  vol.~79, no.~9, pp.~1321--1336, 1991.

\bibitem{rudin1987real}
W.~Rudin, {\em Real and complex analysis}.
\newblock Tata McGraw-Hill Education, 1987.

\bibitem{githubBenchmarks}
``{Piha Benchmarks}.'' https://github.com/PRETgroup/ PihaBenchmarks.
\newblock last accessed - 15.08.2015.

\bibitem{Pedro2005}
H.~Joao~Pedro, ``{How to describe a hybrid system? Formal models for hybrid
  system}.'' University of California at Santa Barbara, Course ECE229, Lecture
  2. http://www.ece.ucsb.edu/\symbol{126}hespanha/ece229/Lectures/
  Lecture2.pdf, 2005.

\bibitem{lygeros2008hybrid}
J.~Lygeros, C.~Tomlin, and S.~Sastry, ``Hybrid systems: Modeling, analysis and
  control.'' University of California, Berkeley, Course EE291e.
  http://inst.cs.berkeley.edu/\symbol{126}ee291e/sp09/handouts/book.pdf, 2009.

\bibitem{chen201487}
T.~Chen, M.~Diciolla, M.~Kwiatkowska, and A.~Mereacre", ``{Quantitative
  verification of implantable cardiac pacemakers over hybrid heart models },''
  {\em Information and Computation}, vol.~236, pp.~87 -- 101, 2014.

\bibitem{platzer2012logical}
A.~Platzer, ``Logical analysis of hybrid systems,'' in {\em Descriptional
  Complexity of Formal Systems}, pp.~43--49, Springer, 2012.

\bibitem{Costello2013}
C.~Brennon and E.~Joshua, ``{Hybrid Systems}.'' Tufts University, Course EE194,
  Lecture. http://www.eecs.tufts.edu/\symbol{126}khan/Courses/Spring2013/EE194/
  Lecs/ Hybrid{\_} Systems\_Presentation\_Elliott\_Costello.pdf, 2013.

\end{thebibliography}

%

\begin{IEEEbiography}
	[{\includegraphics[width=1in,height=1.25in,clip,keepaspectratio]
		{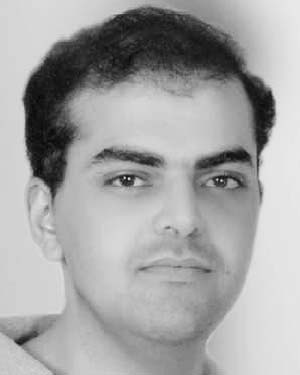}}]
	{Avinash Malik}
Avinash Malik is a Lecturer at the University
of Auckland, New Zealand. His main
research interest lies in programming languages
for multicore and distributed systems
and their formal semantics/ compilation. He
has worked at organizations such as INRIA in
France, Trinity College Dublin, IBM research
Ireland, and IBM Watson on design and compilation
of programming languages. He holds
B.E. and PhD degrees from the University of
Auckland.
\end{IEEEbiography}

\begin{IEEEbiography}
[{\includegraphics[width=1in,height=1.25in,clip,keepaspectratio]
	{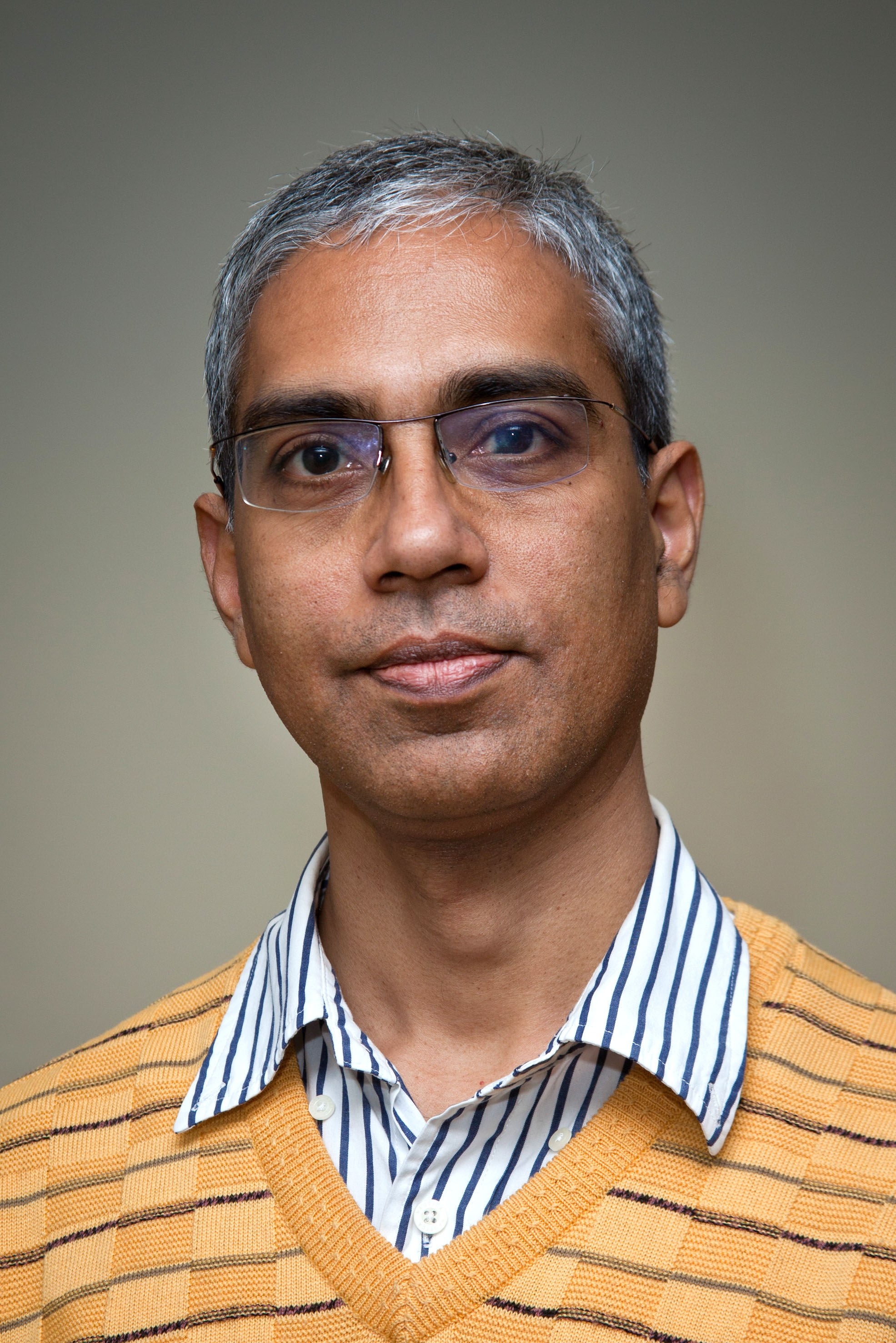}}]
	{Partha S Roop}
	Partha received PhD degree in computer science (software
	engineering) from the University of New south Wales,
	Sydney, Australia, in 2001. He is currently an Associate
	Professor and is the Director of the Computer Systems
	Engineering Program with the Department of Electrical
	and Computer Engineering, University of Auckland, New
	Zealand. 
	Partha is an associated team member of the SPADES team
	INRIA, Rhone-Alpes and held a visiting positions in
	CAU Kiel and  Iowa state University.
	His research interests include the design and
	verification of embedded systems. In particular, he is
	developing techniques for the design of embedded
	applications in automotive, robotics, and intelligent
	transportation systems that meet functional safety
	standards.
\end{IEEEbiography}


\begin{IEEEbiography}
	[{\includegraphics[width=1in,height=1.25in,clip,keepaspectratio]
		{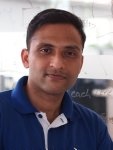}}]
	{Sidharta Andalam}

	Sidharta Andalam received his Ph.D from University of
	Auckland in 2013, where his thesis focused on developing a
	predictable platform for safety-critical systems. He is
	currently a research fellow in embedded systems at the
	University of Auckland, New Zealand. His principle research
	interests are on design, implementation and analysis of
	safety-critical applications. He has worked at TUM CREATE,
	Singapore exploring safety-critical applications in
	automotive domain.

\end{IEEEbiography}

\begin{IEEEbiography}
	[{\includegraphics[width=1in,height=1.25in,clip,keepaspectratio]
		{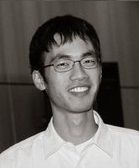}}]
	{Eugene Yip}
	
	Eugene Yip received the B.E. (Hons) and Ph.D degrees in
	electrical and computer systems engineering from the
	University of Auckland, New Zealand. At the start of
	2015, he joined the Heart-on-FPGA research group in the
	Department of Electrical and Computer Engineering,
	University of Auckland, as a research assistant working
	on the real-time modeling of cardiac electrical
	activity. In the middle of 2015, he joined the SWT
	research group at the University of Bamberg as a
	research assistant working on synchronous languages. His
	current research interests include synchronous
	mixed-criticality systems, parallel programming, static
	timing analysis, formal methods for organ modeling, and
	biomedical devices.

\end{IEEEbiography}

\begin{IEEEbiography}
	[{\includegraphics[width=1in,height=1.25in,clip,keepaspectratio]
		{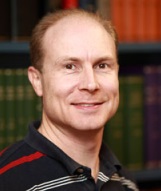}}]
	{Mark Trew}
	
	Mark Trew attained a Bachelor of Engineering in
	Engineering Science in 1992 and a PhD Engineering
	Science in 1999, both from the University of Auckland.
	He is currently a Senior Research Fellow at the Auckland
	Bioengineering Institute. Mark constructs computer
	models and analysis tools for interpreting and
	understanding detailed images of cardiac tissue and
	cardiac electrical activity.

\end{IEEEbiography}

\end{document}